\documentclass{ws-rmta}
\newcommand{\e}{{\rm e}}
\newcommand{\dd}{{\rm d}}

\makeatletter
\let\trimmarks\relax % disables \@begindvi\trimmarks
\makeatother

\begin{document}

\markboth{Huang and Wei}
{Instructions for Typing Manuscripts (Paper's Title)}

\catchline{}{}{}{}{}
%
%%%%%%%%%%%%%%%%%%%%%%%%%%%%%%%%%%%%%%%%%%%%%%%%%%%%%%%%%%%%%%%%%%%%

\title{Non-intersecting Squared Bessel Process: Spectral Moments and Dynamical Entanglement Entropy}

\author{Youyi Huang}

\address{University of Central Missouri\\
Warrensburg, Missouri 64093, USA\,\\ 
\email{yhuang@ucmo.edu} }

\author{Lu Wei}

\address{Texas Tech University\\
Lubbock, Texas 79409, USA\\
\email{luwei@ttu.edu} }

\maketitle

% \begin{history}
% \received{(Day Month Year)}
% \revised{(Day Month Year)}
% \end{history}

\begin{abstract}
Statistical ensembles of reduced density matrices of bipartite quantum systems play a central role in entanglement estimation, but do not capture the non-stationary nature of entanglement relevant to realistic quantum information processing. To address this limitation, we propose a dynamical extension of the Hilbert-Schmidt ensemble, a baseline statistical model for entanglement estimation, arising from non-intersecting squared Bessel processes and perform entanglement estimation via average entanglement entropy and quantum purity. The investigation is enabled by finding spectral moments of the proposed dynamical ensemble, which serves as a new approach for systematic computation of entanglement metrics. Along the way, we also obtain new results for the underlying multiple orthogonal polynomials of modified Bessel weights, including structure and recurrence relations, and a Christoffel-Darboux formula for the correlation kernels.   
\end{abstract}


%\ccode{Mathematics Subject Classification 2000: 22E46, 53C35, 57S20}

\section{Introduction}
Quantifying entanglement is a central problem in quantum information theory, with direct implications for quantum communication, computation, and many-body physics. For bipartite quantum systems in random pure states, the reduced density matrix of a subsystem exhibits highly nontrivial spectral fluctuations, characterized by strong eigenvalue correlations, which are naturally described within random matrix theory~\cite{Mehta,Forrester,Majumdar, MajumdarVivo2012}. In Page’s seminal work~\cite{Page93}, the eigenvalues of reduced density matrices follow the Hilbert-Schmidt ensemble (fixed-trace Wishart-Laguerre ensemble), where the degree of entanglement is estimated via cumulants of entanglement metrics such as entanglement entropy~\cite{Page93,Foong94,Ruiz95,VPO16,Wei17,Wei20,HWC21,HW25}, quantum purity~\cite{Lubkin78, ZyczkowskiSommers2001,Giraud07}, and entanglement capacity~\cite{Okuyama21, Wei23C}. Similar statistical formulations for other random state ensembles have been considered in the literature, notably the Bures-Hall ensemble~\cite{SZ04,OSZ10, Bortola09,Bortola10,Borot12,Bortola14,FK16, SK19, Wei20BHA,Wei20BH,LW21,WHW25BH} and fermionic Gaussian ensemble~\cite{BHK21,BHKRV22,HW22,HW23,HW23C}. Despite substantial progress in those statistical settings, such results are inadequate for realistic quantum systems, where entanglement is inherently non-stationary. Studies of entanglement estimation under real-world scenario of time-dependent stochastic models remain relatively limited. Time-dependent dynamical random matrix ensembles have, however, been extensively studied including most notably the Dyson Brownian motion and its extensions~\cite{Dyson1962,BK04, KT04,KT07,KZK08,KMW09,Forrester, KT2011,DelvauxKuijlaarsZhang2011Tacnode}. Dynamical random matrix framework provides a natural and well-defined mechanism that is well suited for introducing time-dependent deformations into statistical quantum state models to investigate the evolution of quantum entanglement. Understanding this behavior helps identify statistical features of entanglement that are relevant for realistic quantum applications, where entanglement evolves over time during state preparation and algorithm execution.

In this work, we propose a dynamical spectral model generated by non-intersecting squared Bessel processes~\cite{{KT04,KT07,KZK08,KMW09,KT2011}}, which constitutes a dynamical generalization of the Hilbert-Schmidt ensemble and reproduces it in a well-defined limit. The construction introduces genuine time dependence while remaining analytically tractable for exact computation. The main results are the explicit formulas of average purity and average entanglement entropy that characterize the typical entanglement behavior of the considered dynamical ensemble and recover the corresponding results for the Hilbert-Schmidt ensemble in the same limit.

Technically, rather than relying on direct, case-by-case computations~\cite{Foong94,Ruiz95,SZ04, VPO16,Wei17,SK19,Wei20,Wei20BHA,Wei20BH,HWC21,LW21, HW22,HW23,Wei23C, HW23C,HW25, WHW25BH}, these results follow from new recurrence relations for spectral moments derived in this work. Spectral moments is a fundamental object in random matrix theory~\cite{HarerZagier1986,HaagerupThorbjornsen2003,Ledoux04EJP,Ledoux2009GOE,DubrovinYang2017,CundenDahlqvistOConnell21, GisonniGravaRuzza20, CundenDahlqvistOConnell21, ForresterRahman21, ForresterEtAl23qPearson, Deitmar23,YangZhou2023Dessins, Byun24Elliptic, ByunForrester24Ginibre, ByunForresterOh24qGUE, AkemannByunOh25, ByunJungOh25, ForresterNishigaki26}, encoding global information about eigenvalue distributions and underpinning many classical results such as Wigner's semicircle law~\cite{Wigner1958Distribution}, the Marchenko-Pastur law for sample covariance matrices~\cite{MarchenkoPastur1967}, genus expansions and map enumeration arising from exact trace moment formulas~\cite{HarerZagier1986}, free convolution limits and asymptotic freeness in free probability~\cite{Voiculescu1991}, central limit theorems for linear eigenvalue statistics~\cite{Johansson1998}, and strong convergence results for random matrices~\cite{HaagerupThorbjornsen2003,Ledoux2009GOE}. Spectral moments often satisfy exact recurrence relations, see for example, the prototypical cases in~\cite{HarerZagier1986,HaagerupThorbjornsen2003,Ledoux2009GOE}. These recursions provide a powerful tool for both finite-size computations and asymptotic analysis of random matrix ensembles. While existing works have mostly focused on spectral moments of integer order, we extend this framework to real orders, naturally leading to logarithmic linear statistics relevant to entanglement entropy. This approach builds on the framework recently developed in~\cite{HW25}, where spectral moments serve as fundamental building blocks for the systematic derivation of higher-order cumulants of entanglement entropy in the fixed-trace Wishart–Laguerre ensemble.

\subsection{Non-intersecting squared Bessel process}
% By the Karlin–McGregor theorem~\cite{KM}, the joint probability density of the
% positions $x_1,\dots,x_m$ at a fixed time $t\in(0,T)$ of $n$
% non-intersecting paths is a biorthogonal ensemble~\cite{Borodin}
% \begin{equation}\label{eq:bi}
% \frac{1}{Z}\det\left[f_j(x_k)\right]_{1\leq j, k \leq n}
%                \det\left[g_j(x_k)\right]_{1\leq j, k \leq n},
% \qquad x_i\in[0,\infty).
% \end{equation}
% Here $Z>0$ denotes the normalization constant and 
% \begin{equation}
% f_j(x) = \widehat{P}_t(a_j,x), \qquad
% g_j(x) = \widehat{P}_{T-t}(x,b_j), \qquad 0<t<T<\infty,
% \end{equation}
% are given in terms of the transition probability density $\widehat{P}_t(x,y)$ of the underlying diffusion with endpoints 
% \begin{equation}\label{eq:condition}
% a_1 < a_2 < \cdots < a_n, \qquad b_1 < b_2 < \cdots < b_n.
% \end{equation} 
We introduce here the formulation that leads to the non-intersecting squared Bessel process~\cite{CV03,CV032,KT04,KT07,KZK08,ForresterNagao2008,KMW09,Forrester,KT2011,DelvauxKuijlaarsZhang2011,DelvauxKuijlaarsZhang2011,Katori2016Book,LiuYaoZhang25HardEdgeTacnode}. Let $X(t)$ be a one-dimensional diffusion with transition probability density $\widehat P_t(x,y)$. By the Karlin-McGregor theorem~\cite{KM}, the joint probability density of positions $(x_1,\dots,x_n)=(X_1(t),\dots,X_n(t))$ of $n$ independent copies $X_j$, $j=1,\dots,n$ with endpoints
\begin{equation}\label{eq:condition}
X_j(0)=a_j,\qquad X_j(T)=b_j,\qquad 
a_1<\cdots<a_n,\ \ b_1<\cdots<b_n,
\end{equation}
and conditioned to never intersect over $(0,T)$, is a biorthogonal ensemble~\cite{Borodin}
\begin{equation}\label{eq:bi}
\frac{1}{Z}\det\left[f_j(x_k)\right]_{1\leq j, k \leq n}
               \det\left[g_j(x_k)\right]_{1\leq j, k \leq n},
\qquad 0\le x_1<\cdots<x_n.
\end{equation}
Here $Z>0$ denotes the normalization constant and 
\begin{equation}
f_j(x) = \widehat{P}_t(a_j,x), \qquad
g_j(x) = \widehat{P}_{T-t}(x,b_j), \qquad 0<t<T<\infty.
\end{equation}

The joint probability density of the non-intersecting squared Bessel paths is obtained by taking $\widehat{P}_t(x,y)$ the transition density of a squared Bessel process
\begin{equation}\label{eq:bpb}
\widehat{P}_t(x,y)=
\frac{1}{2t}
\left(\frac{y}{x}\right)^{\alpha/2}
\e^{-\tfrac{1}{2t}(x+y)}I_{\alpha}\!\left(\frac{\sqrt{xy}}{t}\right), \qquad x, y>0,~~~\alpha>-1
\end{equation}
with
\begin{equation}\label{eq:Iv}
    I_\nu(z) = \left( \tfrac{1}{2} z \right)^\nu \sum_{i=0}^\infty \frac{\left( \tfrac{1}{4} z^2 \right)^i}{\Gamma(i+1) \Gamma(\nu + i + 1)}
\end{equation}
denoting the modified Bessel function.
% The normalization constant is 
% \begin{equation}
%  Z=n!\det\left[ \widehat{P}_{T}\!\left(a_i, b_j\right)\right].  
% \end{equation}
% with
% \begin{eqnarray}
%  z_{ij}&=&\int_0^\infty \widehat{P}_t\!\left(a_i, x\right) \widehat{P}_{T-t}\!\left(x, b_j\right)\dd x  \\ &=& \widehat{P}_{T}\!\left(a_i, b_j\right),
% \end{eqnarray}
% where the second equality is known as the Chapman–Kolmogorov property.
\begin{figure}[h]
\centerline{\psfig{file=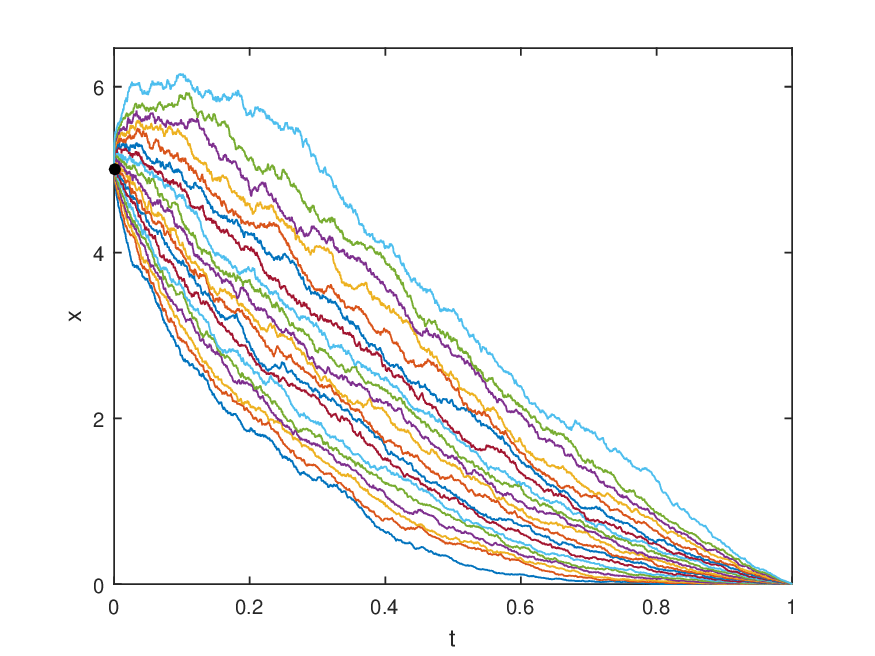,width=4.0in}}
\vspace*{8pt}
\caption{Numerical simulation of 20 rescaled non-intersecting squared Bessel paths with $a = 5$, $T=1$.}\label{fig1}
\end{figure}
The present work focuses on the case~\cite{KMW09}
\begin{equation}\label{eq:confluent}
a_1 = \cdots = a_n \to a>0, \qquad b_1 = \cdots = b_n \to 0,
\end{equation}
where the joint probability density~(\ref{eq:bi}) is specified by
\begin{eqnarray}\label{eq:bifg}
    f_{2j-1}(x) = x^{j-1} w_1(x), \qquad f_{2j}(x) = x^{j-1} w_2(x), \qquad g_j(x)     = x^{j-1},
\end{eqnarray}
with weights 
\begin{eqnarray}
w_1(x)&=&x^{\frac{\alpha}{2}}\e^{-\frac{T}{2t(T-t)} x}\, I_{\alpha}\!\left(\frac{\sqrt{a\,x}}{t}\right),\label{eq:we1}\\
w_2(x)&=&x^{\frac{\alpha+1}{2}}\e^{-\frac{T}{2t(T-t)} x}\, I_{\alpha+1}\!\left(\frac{\sqrt{a\,x}}{t}\right).\label{eq:we2}
\end{eqnarray}
See Figure \ref{fig1} for a numerical simulation of the case (\ref{eq:confluent}) with $a = 5$.
% \begin{figure}[pb]
% \centerline{\psfig{file=1500.eps,width=4.0in}}
% \vspace*{8pt}
% \caption{Numerical simulation of 20 rescaled non-intersecting Squared Bessel paths with a = 0.}
% \end{figure}

In the limit $a \to 0$, the biorthogonal structure~(\ref{eq:bi}) further collapses to the Wishart-Laguerre with a deterministic rescaling in the exponential weight,
\begin{equation}\label{eq:gwi}
\frac{1}{Z}\prod_{i=1}^{n}x_i^\alpha \e^{-\frac{T}{2t(T-t)} x_i} \prod_{1\leq i<j\leq n}(x_i-x_j)^2,
\end{equation}
whose correlation kernel is expressed in terms of Laguerre polynomials. Note that despite the rescaled Laguerre weight, the ensemble is not genuinely time-dependent after trace normalization. The corresponding fixed-trace ensemble remains Hilbert-Schmidt ensemble~(\ref{eq:hs}), and the entanglement entropy can be computed exactly as in the standard Wishart-Laguerre case, cf.~\cite{Page93,Ruiz95,VPO16,Wei17,Wei20,HWC21,HW25}.  

The considered regime~(\ref{eq:confluent}) therefore constitutes a simplest yet non-trivial time-dependent generalization of Page’s setting~\cite{Page93}: It preserves the static structure in the limit $a\to 0$, while introducing a dynamical deformation through the stochastic evolution of the
non-intersecting paths.

%Note that polynomial observables such as quantum purity can be accessed through integer-order spectral moments,
% with the second moment yielding the average purity, whereas logarithmic observables
% have often required more explicit, case-by-case analysis. Recent developments~\cite{HW25} have highlighted the importance of spectral moments with real indices. Such generalized moments, in particular the corresponding recurrence relations, provide a unified framework where logarithmic linear statistics such as entanglement entropy can be accessed by differentiation with respect to the moment index. The resulting identities serve as essential building blocks for the systematic computation of higher-order entanglement fluctuations. In the present work, we adopt this spectral-moment perspective.

% The determinantal structure of LUE makes it possible to compute exact expressions for cumulants of entanglement metrics, such as purity, entanglement entropy, and entanglement capacity, that reflect the statistical behavior of degree of entanglement.

The necessary definitions are presented first before moving on to the main text. We denote two families of linear statistics over the non-intersecting squared Bessel process (\ref{eq:bi}) conditioned at (\ref{eq:confluent}), respectively being  
\begin{equation}\label{eq:Rk}
R_k = \sum_{i=1}^{n} x_i^k,
\end{equation}
and
\begin{equation}\label{eq:Tk}
T_k = \sum_{i=1}^{n} x_i^k\ln x_i,
\end{equation}
where $k\geq 0$  and we omit the subscript when $k=1$. The $k$-th spectral moment is defined by
\begin{equation}\label{eq:mk}
m_k=\mathbb{E}[R_k],
\end{equation}
where the expectation is taken over the non-intersecting squared Bessel process~(\ref{eq:bi}) as specified in (\ref{eq:bifg}).

The rest of the paper is structured as follows. In Section~\ref{sec:2}, we derive the spectral moments as defined above. Section~\ref{sec:bo} presents both existing and new results on biorthogonal polynomials associated with non-intersecting squared Bessel processes, which are required for establishing the recurrence relations of the spectral moments in Section~\ref{sec:rr}. In Section~\ref{sec:ee}, we perform entanglement estimation of the dynamical model using the findings developed in Section~\ref{sec:2}, with the main results on the average quantum purity and the average entanglement entropy derived in Sections~\ref{sec:p} and~\ref{sec:v}, respectively.

\section{Spectral Moments of Non-intersecting Squared Bessel Process}\label{sec:2}
We investigate spectral moments, which serve as the key step in the current framework of computing average entanglement metrics in Section~\ref{sec:ee}. In particular, we derive recurrence relations of the spectral moments, stated in Proposition~\ref{prop:r1} and Proposition~\ref{prop:r2}. The central idea in establishing these recurrence relations is the repeated use of integration by parts. This allows the spectral moments to be expressed in terms of integrals involving orthogonal polynomials, which are then systematically eliminated by deriving additional compatible relations through further applications of integration by parts. The entire procedure relies fundamentally on the recurrence and structure relations of the underlying biorthogonal polynomials, summarized in Propositions~\ref{propQ} and~\ref{propstruP}, as well as on the Christoffel-Darboux formula for the correlation kernels in Proposition~\ref{prop:cd} and its derivative form given in Proposition~\ref{prop:cdd}. These necessary polynomial and kernel-level results are first discussed in Section~\ref{sec:bo}, prior to the derivation of the recurrence relations in Section~\ref{sec:rr}.

\subsection{Determinantal process of non-intersecting squared Bessel paths: Existing and new results}\label{sec:bo}
The biorthogonal structure~(\ref{eq:bi}) ensures the associated point process is determinantal: The $l$-point correlation function for every $1\leq l\leq n$ can be written as 
\begin{equation}\label{eq:gl}
   g_l(x_1,\dots,x_l)=\det\left[K(x_i,x_j)\right]_{i,j=1}^{l},
\end{equation}
the correlation kernel is given in~\cite{KMW09} as
\begin{equation}\label{eq:K-MOP}
  K(x,y) = \sum_{j=0}^{n-1} Q_j(x)P_j(y),
\end{equation}
where 
\begin{equation}\label{eq:qj}
Q_{j}(x)=A_{j,1}(x)w_1(x)+A_{j,2}(x)w_2(x)
\end{equation}
with $A_{j,1}$ and $A_{j,2}$ being polynomials of degrees $\lceil j/2 \rceil$ and $\lfloor j/2 \rfloor$, respectively, and $P_j$ is a monic polynomial of degree $j$. They satisfy the biorthogonality property
\begin{equation}\label{eq:orth}
    \int_{0}^{\infty}Q_j(x) P_i(x) \dd x =\delta_{ji},\qquad j,i=1,\dots,n-1.
\end{equation}
The orthogonality conditions of $Q_{j}$ and $P_{j}$ that ensures the above property are 
\begin{eqnarray}
    \int_0^\infty Q_j(x)x^i \dd x&=& 0,\qquad i=0,1,\dots,j-1,\label{eq:orth1}\\
    \int_0^\infty Q_j(x)x^j \dd x&=& 1\label{eq:orth2}
\end{eqnarray}
and
\begin{eqnarray}
    \int_0^\infty P_j(x)x^iw_1(x)  \dd x&=& 0, \qquad i=0,1,\dots, \left\lceil \frac{j}{2} \right\rceil-1,\label{eq:orth3}\\
    \int_0^\infty P_j(x)x^iw_2(x)  \dd x&=& 0, \qquad i=0,1,\dots, \left\lfloor \frac{j}{2} \right\rfloor-1\label{eq:orth4},
\end{eqnarray}
respectively. The single index \(n\) of the polynomials encodes the underlying multi-index \((j_1,j_2)\) via
\(j = j_1 + j_2\), with \(j_1 = \lceil j/2\rceil\) and \(j_2 = \lfloor j/2\rfloor\). The polynomials $A_{j,1}$ and $A_{j,2}$ are called multiple orthogonal polynomials of type~I, and $P_j$  the multiple orthogonal polynomials of type~II.

\subsubsection{Existing results on biorthogonal polynomials}
The multiple orthogonal polynomials associated with weights
\begin{eqnarray}
w_{\alpha,c}(x)=x^{\frac{\alpha}{2}}\e^{-c x}  I_\alpha\left(2 \sqrt{x}\right), \qquad  w_{\alpha+1,c}(x)=x^{\frac{\alpha+1}{2}}\e^{-c x}  I_\alpha\left(2 \sqrt{x}\right) \label{eq:wea}
\end{eqnarray}
that share the same functional form as those in~(\ref{eq:we1})--(\ref{eq:we2}) have been studied by Coussement and Van Assche in~\cite{CV03}. 
Let
\(\widetilde{A}_{n,1}^{(\alpha,c)}\) and \(\widetilde{A}_{n,2}^{(\alpha,c)}\) denote the type~I polynomials, and 
\begin{equation}
\widetilde{Q}_{n}^{(\alpha,c)}(x)=\widetilde{A}_{n,1}^{(\alpha,c)}(x)w_{\alpha,c}(x)+\widetilde{A}_{n,2}^{(\alpha,c)}(x)w_{\alpha+1,c}(x),
\end{equation}
while $\widetilde{P}_{n}^{(\alpha,c)}$ denotes the associated type~II polynomials.   They are characterized by the orthogonality conditions
\begin{equation}
    \int_0^\infty \widetilde{Q}_n^{(\alpha,c)}(x)\,x^j \,\dd x = 0,
    \qquad j=0,1,\dots,n-1, \label{eq:orthot1}
\end{equation}
and
\begin{align}
    \int_0^\infty \widetilde{P}_n^{(\alpha,c)}(x)\,x^j w_{\alpha,c}(x)\, \dd x &= 0,
        & j&=0,1,\dots, \Big\lceil \tfrac{n}{2} \Big\rceil -1,\label{eq:orthot2}\\
    \int_0^\infty \widetilde{P}_n^{(\alpha,c)}(x)\,x^j w_{\alpha+1,c}(x)\, \dd x &= 0,
        & j&=0,1,\dots, \Big\lfloor \tfrac{n}{2} \Big\rfloor -1.\label{eq:orthot3}
\end{align}

We summarize here the existing results established in~\cite{CV03} that are revelent to our work, the corresponding statements for $P_n$ and $Q_n$ then follow by a rescaling \begin{eqnarray}
Q_n(x)&=&\frac{1}{\widetilde{H}_n}\left(\frac{\beta}{c}\right)^{n+1}\widetilde{Q}_n^{(\alpha,c)}\!\left(\frac{\beta}{c}x\right)\label{eq:r1}\\
P_n(x)&=&\left(\frac{\beta}{c}\right)^{-n}\widetilde{P}_n^{(\alpha,c)}\!\left(\frac{\beta}{c}x\right)\label{eq:r2},
\end{eqnarray} 
where
\begin{eqnarray}
 c&=&\frac{T}{2t(T-t)}\frac{4t^2}{a}\label{eq:c}\\
\beta&=&\frac{T}{2t(T-t)},
\end{eqnarray}
and $\widetilde{H}_k$ denotes the normalization constant for $ \widetilde{Q}_n^{(\alpha,c)}$,
\begin{equation}
        \widetilde{H}_k=\int_0^{\infty} x^{n} \widetilde{Q}_n^{(\alpha,c)}(x)\, \dd x,
\end{equation}
which, by applying (\ref{eq:xqk}) with the specialization $s = n+1$, admits the explicit expression 
\begin{equation}
 \widetilde{H}_k= (-1)^n \Gamma(n+1)\, \e^{1/c} c^{-\alpha - n - 1}.    
\end{equation}

% The following relations provide an explicit conversion between the two families of multiple orthogonal polynomials, allowing results for one family to be directly translated to the other. We have

% The relations are obtained as follows. Comparing with the orthogonality conditions~(\ref{eq:orth1})--(\ref{eq:orth4}), a change of variables
% \begin{equation}
%     y\to \frac{a}{4t^2} x
% \end{equation}
% with $c$ being assigned as (\ref{eq:c}) in (\ref{eq:orthot1})--(\ref{eq:orthot3}) yields the relations~(\ref{eq:r1})--(\ref{eq:r2}). 

\subsubsection*{- Rodrigues' formula} The type~I multiple orthogonal polynomials are determined by (\ref{eq:orthot1}) up to a multiplicative constant.  
Fixing this constant to be \(1\), the functions $\widetilde{Q}_n^{(\alpha,c)}$ satisfy the Rodrigues' formula 
 \begin{equation}\label{eq:rdq}
     \widetilde{Q}_n^{(\alpha,c)}(x)=\frac{\dd ^n}{\dd x^n}w_{\alpha +n,c}(x).
 \end{equation}

\subsubsection*{- Recurrence relations}
The 3-point relation provides a parameter-lowering relation in $\alpha$ for $\widetilde{Q}_{n}^{(\alpha ,c)}$ is
\begin{equation}\label{eq:3rql}
    \widetilde{Q}_{n}^{(\alpha+1,c)}(x)=\frac{1}{c}\widetilde{Q}_{n}^{(\alpha,c)}(x)-\frac{1}{c}\widetilde{Q}_{n+1}^{(\alpha ,c)}(x).
\end{equation}
The type II multiple orthogonal polynomials $\widetilde{P}^{(\alpha,c)}_n(x)$ admit the four-point recurrence relation
\begin{eqnarray}\label{eq:recurp}
x \widetilde{P}^{(\alpha,c)}_n(x)&=&\widetilde{P}^{(\alpha,c)}_{n+1}(x)+\frac{c (2 n+\alpha+1)+1 }{c^2}\widetilde{P}^{(\alpha,c)}_n(x)+\frac{n (c (\alpha +n)+2)}{c^3} \widetilde{P}^{(\alpha,c)}_{n-1}(x)\nonumber\\
&&+\frac{n (n-1) }{c^4}\widetilde{P}^{(\alpha,c)}_{n-2}(x).
\end{eqnarray}

With (\ref{eq:r2}) we reproduce the recurrence relation of type II multiple orthogonal polynomials $P_n$ provided in~\cite{KMW09},
\begin{eqnarray}
x P_{n}(x)&=&P_{n+1}(x)+\frac{c (\alpha +2 n+1)+1 }{\beta  c}P_n(x)+\frac{n (c (\alpha +n)+2) }{\beta ^2 c}P_{n-1}(x)\nonumber\\
&&+\frac{n(n-1) }{\beta ^3 c}P_{n-2}(x),\label{eq:recurP}
\end{eqnarray}
from the recurrence relation~(\ref{eq:recurp}) of $\widetilde{P}_n^{(\alpha,c)}$.

 \subsubsection*{- Explicit formulas}
 The explicit expressions for $\widetilde{Q}^{(\alpha,c)}_n$ and $\widetilde{P}^{(\alpha,c)}_n$ are respectively given by
\begin{equation}
    \widetilde{Q}_{n}^{(\alpha,c)}(x)=\sum _{k=0}^n (-c)^k \binom{n}{k} w_{k+\alpha,c}(x),
\end{equation}
 and
\begin{eqnarray}
\widetilde{P}^{(\alpha,c)}_n(x)&=&\frac{(-1)^n }{c^{2 n}}\sum _{k=0}^n c^k (n-k+1)_k L_k^{(\alpha)}(c x).\label{eq:pklk}\\
&=&\sum _{k=0}^n \frac{\Gamma (n+1)}{\Gamma (k+1)} c^{k-n} L_{n-k}^{(-n-\alpha -1)}\!\left(\frac{1}{c}\right)x^k.\label{eq:pklkx}
\end{eqnarray}
Here $L_k^{(\alpha)}$ denotes the Laguerre polynomials, which admit the explicit expressions
\begin{eqnarray}
L_k^{(\alpha )}(x)  
&=&\sum_{j=0}^{k}(-1)^j \binom{k+\alpha}{k-j}\frac{x^j}{j!},
  \\
&=& \frac{(\alpha +1)_k}{n!} \, _1F_1(-k;\alpha +1;x).
\end{eqnarray}

\subsubsection*{- Differential properties}
The multiple orthogonal polynomials also admit the following differential properties,
\begin{equation}\label{eq:dfqhat1}
\frac{d}{dx} \widetilde{Q}_n^{(\alpha+1,c)}(x) = \widetilde{Q}_{n+1}^{(\alpha,c)}(x),
\end{equation}
\begin{equation}\label{eq:dfphat1}
  \frac{\dd}{\dd x} \widetilde{P}_n^{(\alpha,c)}(x)=\widetilde{P}_{n-1}^{(\alpha+1,c)}(x).  
\end{equation}

We also note that the results collected here are not exhaustive. For interested readers, we refer to~\cite{CV03,CV032} for a more complete account.

\subsubsection{New results on biorthogonal polynomials of non-intersecting squared Bessel process}
In this section, we present our new results on the multiple orthogonal polynomials and the associated correlation kernels for the biorthogonal ensemble arising from the non-intersecting squared Bessel process conditioned at~(\ref{eq:condition}). We first establish the fundamental recurrence and structure relations for the type~I functions $Q_n$, summarized in Proposition~\ref{propQ}, together with the recurrence relation for the type~II polynomials $P_n$ given in Proposition~\ref{propstruP}. These relations provide the basic algebraic framework for the underlying biorthogonal structure and are essential for performing kernel manipulations in a systematic way. Building on them, we then derive the Christoffel-Darboux formula for the correlation kernel and its derivative identities, given in Propositions~\ref{prop:cd} and~\ref{prop:cdd}, respectively. To the best of our knowledge, all of these results are new.

%  We now collect the relevant results for $ \widetilde{Q}_{k}$ and  $ \widetilde{P}_{k}$ established in~\cite{CV03} before introducing the new results. 
%The corresponding statements for $Q_k$ and $P_k$ associated with the weight are then obtained by a simple rescaling.

% The normalization constant will be used in Lemma~\ref{lemmapr} to establish the relations between the two families of multiple orthogonal polynomials.

\begin{lemma}\label{lemmarrq}
The type I function $\widetilde{Q}_n^{(\alpha,c)}$ admits the four-point rule
\begin{equation}
x\widetilde{Q}_n^{(\alpha,c)}(x)=-n \widetilde{Q}_{n-1}^{(\alpha+1,c)}(x)-\frac{c (\alpha +n+1)+1 }{c}\widetilde{Q}_n^{(\alpha+1,c)}(x)-\frac{1}{c}\widetilde{Q}_{n+1}^{(\alpha+1,c)}(x).
\end{equation}
\end{lemma}
\begin{proof}
By~(\ref{eq:3rql}), we can have
\begin{eqnarray}
    &&-\frac{1}{c}\widetilde{Q}_{n+1}^{(\alpha+1,c)}(x)+\frac{c (\alpha +n+1)+1}{c}\widetilde{Q}_{n}^{(\alpha+1,c)}(x)-n\widetilde{Q}_{n-1}^{(\alpha+1,c)}(x)\nonumber\\
    &=&\frac{1}{c^2}\widetilde{Q}_{n+2}^{(\alpha,c)}(x)-\frac{c (\alpha +n+1)+2}{c^2} \widetilde{Q}_{n+1}^{(\alpha,c)}(x)+\frac{c (\alpha +2 n+1)+1 }{c^2}\widetilde{Q}_n^{(\alpha,c)}(x)\nonumber\\
    &&-\frac{n }{c}\widetilde{Q}_{n-1}^{(\alpha,c)}(x).\label{eq:lemmarrq1}
\end{eqnarray}
The proof of Lemma~\ref{lemmarrq} is completed by simplifying the right-hand side of~(\ref{eq:lemmarrq1}) to $x\,\widetilde{Q}_{m}^{(\alpha,c)}(x)$ using the four-point recurrence relation derived below in~(\ref{eq:recurq}).
\end{proof}

\begin{proposition}\label{propQ}
The recurrence relation and structure relation of type I functions $Q_n$ in (\ref{eq:qj}) are respectively given by 
\begin{eqnarray}\label{eq:recurQ}
 xQ_n(x)&=&\frac{(n+1) (n+2) }{\beta ^3 c}Q_{n+2}(x)+ \frac{(n+1) (c (\alpha +n+1)+2)}{\beta ^2 c} Q_{n+1}(x)\nonumber\\
 &&+\frac{c (\alpha +2 n+1)+1 }{\beta  c}Q_n(x)+Q_{n-1}(x),  
\end{eqnarray}
and
\begin{equation}\label{eq:structQ}
x\frac{\dd}{\dd x}Q_n(x)=(n+1)\!\left(-\frac{n+2 }{\beta ^2 c}Q_{n+2}(x)-\frac{c (\alpha +n+1)+1}{\beta  c} Q_{n+1}(x)-Q_n(x)\right).
\end{equation}
\end{proposition}
\begin{proof}
Using Rodrigues formula (\ref{eq:rdq}) to rewrite the terms in $n$-th derivative of the relation~\cite{CV03},
\begin{equation}
   x \frac{\dd}{\dd x}w_{v,c}(x)=(v-cx)w_{v,c}(x)+w_{v+1,c}(x)
\end{equation}
 for
\begin{equation}
v = n+\alpha+1,
\end{equation}
into type I functions, we arrive at
\begin{equation}
x \widetilde{Q}_{n+1}^{(\alpha ,c)}(x)=\widetilde{Q}_n^{(\alpha +2,c)}(x)-c n \widetilde{Q}_{n-1}^{(\alpha +2,c)}(x)+(\alpha +1) \widetilde{Q}_n^{(\alpha +1,c)}(x)-cx\widetilde{Q}_n^{(\alpha +1,c)}(x). \label{eq:rqd}
\end{equation}
We express all functions in terms of the parameter $\alpha$ only in the above relation (\ref{eq:rqd}) by applying (\ref{eq:3rql}). This gives the four-point recurrence relation 
\begin{eqnarray}\label{eq:recurq}
x \widetilde{Q}_n^{(\alpha,c)}(x)&=&\frac{1}{c^2}\widetilde{Q}_{n+2}^{(\alpha,c)}(x)-\frac{c (\alpha +n+1)+2 }{c^2}\widetilde{Q}_{n+1}^{(\alpha,c)}(x)\nonumber\\
&&+\frac{c (\alpha +2 n+1)+1 }{c^2}\widetilde{Q}_n^{(\alpha,c)}(x)-\frac{n}{c} \widetilde{Q}_{n-1}^{(\alpha,c)}(x).    
\end{eqnarray}
The result (\ref{eq:recurQ}) is thus proved after applying (\ref{eq:r1}). 

We next establish the structure relation~(\ref{eq:structQ}). By using Lemma~\ref{lemmarrq} to express the polynomials with parameter $\alpha$ in terms of the polynomials with parameter $\alpha+1$ in~(\ref{eq:dfqhat1}), we obtain the structure relation for $\widetilde{Q}_n^{(\alpha,c)}$ as
\begin{equation}
x\frac{\dd}{\dd x}\widetilde{Q}_n^{(\alpha,c)}(x)=-\frac{1}{c}\widetilde{Q}_{n+2}^{(\alpha,c)}(x)+\frac{c(\alpha +n+1)+1}{c}\widetilde{Q}_{n+1}^{(\alpha,c)}(x)-\frac{c (n+1) }{c}\widetilde{Q}_n^{(\alpha,c)}(x).
\end{equation}
The claimed relation~(\ref{eq:structQ}) follows by inserting (\ref{eq:r1}). 
\end{proof}

Lemma~\ref{corop} summarizes several additional properties of $\widetilde{P}_n^{(\alpha,c)}$ that complement those established in~\cite{CV03}.
\begin{lemma}\label{corop}
The type II polynomial $\widetilde{P}_n^{(\alpha,c)}$ admits the following relations
 \begin{eqnarray}
&& \widetilde{P}_{n}^{(\alpha,c)}(x)=\frac{m}{c}\widetilde{P}_{n-1}^{(\alpha +1,c)}(x)+\widetilde{P}_{n}^{(\alpha +1,c)}(x),\label{eq:pr1}\\
&&x \widetilde{P}_n^{(\alpha +1,c)}(x)=\frac{n }{c^3}\widetilde{P}_{n-1}^{(\alpha,c)}(x)+\frac{c (n+\alpha+1)+1}{c^2}\widetilde{P}_n^{(\alpha,c)}(x)+\widetilde{P}_{n+1}^{(\alpha,c)}(x),\label{eq:pr2}\\
&&\widetilde{P}_n^{(\alpha +1,c)}(x)=\sum_{j=0}^n \frac{n!}{j!}(-1)^{j+n} c^{j-n} \widetilde{P}_j^{(\alpha,c)}(x).\label{eq:pr3}
 \end{eqnarray}
\end{lemma}

\begin{proof}
We notice that~\label{eq:lkpk}
\begin{equation}
L_n^{(\alpha)}(c x)=(-1)^n\sum _{j=0}^n \frac{c^{2 j-n} }{\Gamma (j+1) \Gamma (n-j+1)}\widetilde{P}_j^{(\alpha,c)}(x),
\end{equation}
by fitting (\ref{eq:pklk}) in the binomial transform
\begin{equation}
A_n = \sum_{j=0}^n \binom{n}{j} B_j
\end{equation}
before applying the inversion formula
\begin{equation}
B_n = \sum_{k=0}^n (-1)^{n-k} \binom{n}{k} A_k.
\end{equation}    
Inserting (\ref{eq:lkpk}) into the 3-point-rule of Laguerre polynomials
\begin{equation}\label{eq:lk3}
L_n^{(\alpha)}(cx)-L_n^{(\alpha+1)}(cx)+L_{n-1}^{(\alpha+1)}(cx)=0,    
\end{equation}
the resulting summations can be combined into one as
\begin{equation}
    \sum_{j=0}^n \frac{c^{2 j-n-1}}{\Gamma (j+1) \Gamma (n-j+1)}\left(j\widetilde{P}_{j-1}^{(\alpha +1,c)}(x)+c \widetilde{P}_{j}^{(\alpha +1,c)}(x)-c \widetilde{P}_{j}^{(\alpha,c)}(x)\right)=0,
\end{equation}
which is valid for any $n$. Hence the relation (\ref{eq:pr1}) holds. 

By (\ref{eq:pr1}), one can have
\begin{eqnarray}
&&\widetilde{P}_{n+1}^{(\alpha-1,c)}(x)+\frac{c (\alpha +n)+1}{c^2}\widetilde{P}_{n}^{(\alpha-1,c)}(x)+\frac{m}{c^3}\widetilde{P}_{n}^{(\alpha-1,c)}(x)\nonumber\\
&=&    \widetilde{P}^{(\alpha,c)}_{n+1}(x)+\frac{c (2 n+\alpha+1)+1 }{c^2}\widetilde{P}^{(\alpha,c)}_n(x)+\frac{n (c (\alpha +n)+2)}{c^3} \widetilde{P}^{(\alpha,c)}_{n-1}(x)\nonumber\\
&&+\frac{n (n-1) }{c^4}\widetilde{P}^{(\alpha,c)}_{n-2}(x).
\end{eqnarray}
Inserting (\ref{eq:recurp}) into the above relation before setting $\alpha\to\alpha+1$, the result (\ref{eq:pr2}) is established. 

The equality~(\ref{eq:pr3}) is also a consequence of~(\ref{eq:pr1}) and can be proved by induction. 
\end{proof}

\begin{proposition}\label{propstruP}
The structure relation of the type II polynomial $P_n$ is given by
    \begin{equation}\label{eq:structP}
      x\frac{\dd}{\dd x}P_n(x)=  n P_n(x)+\frac{n (c (\alpha +n)+1) }{\beta  c}P_{n-1}(x)+\frac{(n-1) n}{\beta ^2 c}P_{n-2}(x).
    \end{equation}
\end{proposition}
\begin{proof}
The result~(\ref{eq:pr2}) allows us to replace the polynomials of parameter $\alpha+1$ with those of $\alpha$ in (\ref{eq:dfphat1}), which  yields the structure relation for $\widetilde{P}_n^{(\alpha,c)}$ as
\begin{equation}\label{eq:strucp}
x \frac{\dd}{\dd x} \widetilde{P}_n^{(\alpha,c)}(x)= \frac{n (n-1)}{c^3}\widetilde{P}_{n-2}^{(\alpha,c)}(x)+\frac{n (c (\alpha +n)+1)}{c^2}\widetilde{P}_{n-1}^{(\alpha,c)}(x)+n\widetilde{P}_{n}^{(\alpha,c)}(x).
\end{equation}
Together with (\ref{eq:r2}), the structure relation~(\ref{eq:structP}) of $P_n$ is established.
\end{proof}

It is worth noting that the established structure relation~(\ref{eq:strucp}) along with the known recurrence relation~(\ref{eq:recurP}) provide an alternative derivation to the four-term differential equation given in~\cite{CV03} for $y = \widetilde{P}_n^{(\alpha,c)}(x)$,
\begin{equation}
-c^2 m y+\left(c^2 x+c (-\alpha +m-2)-1\right) y^\prime +(\alpha -2 c x+2) y^{\prime\prime}+x y^{\prime\prime\prime}=0, 
\end{equation}
where we omit the details.

\subsubsection*{- Christoffel-Darboux formula and derivatives of correlation kernels}
\begin{proposition}\label{prop:cd}
A Christoffel-Darboux formula of correlation kernel~(\ref{eq:K-MOP}) is given by
\begin{eqnarray}\label{eq:cd}
      (x-y) K(x,y)&=&\frac{m (m+1) }{\beta ^3 c}Q_{m+1}(x)P_{m-1}(y)+\frac{m(m-1)  }{\beta ^3 c}Q_m(x)P_{m-2}(y)\nonumber\\
      && +\frac{m (c (\alpha +m)+2) }{\beta ^2 c}Q_m(x)P_{m-1}(y)-Q_{m-1}(x)P_m(y),
\end{eqnarray}
and in the confluent limit $y\to x$, we have
\begin{eqnarray}\label{eq:cdc}
x K(x,x)&=&\frac{n (n+1)}{\beta ^2 c}Q_{n+1}(x)P_n(x)+\frac{n \left(n^2-1\right)}{\beta ^4 c^2} Q_{n+1}(x)P_{n-2}(x) \nonumber\\
&&+\frac{n (c (\alpha +n)+1)}{\beta  c}Q_{n}(x)P_{n}(x)-\frac{n (c (\alpha +n-\beta  x)+1)}{\beta ^2 c^2}Q_{n}(x)P_{n-1}(x)\nonumber\\
&& +\frac{c n (\alpha +n)+n}{\beta  c}Q_{n-1}(x)P_{n-1}(x)+\!\left(\alpha +\frac{1}{c}+2 n-\beta  x\right)\!Q_{n-1}(x)P_{n}(x)\nonumber\\
&&+\frac{(n-1) n}{\beta ^2 c}Q_{n-1}(x)P_{n-2}(x).\label{eq:cd1}
\end{eqnarray}
\end{proposition}
Note that a general Christoffel-Darboux formula for multiple orthogonal polynomials, formulated in terms of multi-indices, is given in~\cite{DK04}. The formula~(\ref{eq:cd}) is proved in the same spirit as~\cite{DK04} yet may not directly follow from the general formula. Furthermore, the special case~(\ref{eq:cdc}) would not be accessible without the structure relations of these polynomials.  

\begin{proof}
Since $P_k(x)$ is a polynomial of degree $k$, we can expand $x P_k(x)$ as 
\begin{equation}\label{eq:lspp}
xP_k(x)=\sum_{j=0}^{k+1}a_{j,k}P_j(x).
\end{equation}
By multiplying both sides of (\ref{eq:lspp}) with $Q_j(x)$ and integrating over the real line, due to orthogonality~(\ref{eq:orth}), the surviving term gives
\begin{equation}\label{eq:coc}
    a_{j,k}=\int_0^{\infty}x P_k(x)Q_j(x)\dd x.
\end{equation}
On the other hand, $x Q_k(x)$ can be written as the linear combination 
\begin{equation}\label{eq:lspq}
x Q_j(x)=\sum_{k=0}^{n+1}a_{j,k}Q_k(x), \qquad j=0,\dots ,n-1.
\end{equation}
Note that the coefficients $a_{j,k}$ are $0$ if $j\geq k+2$.
Using the expansions (\ref{eq:lspp}) and (\ref{eq:lspq}) we can write
\begin{eqnarray}
   && (y-x)\sum_{k=0}^{n-1}Q_k(x)P_k(y)\label{eq:sumhatpq}
    \\&=&\sum_{k=0}^{n-1}\sum_{j=0}^{k+1}a_{j,k}Q_k(x)P_j(y)-\sum_{k=0}^{n-1}\sum_{j=0}^{n+1}a_{k,j}Q_j(x)P_k(y)\\
    &=&\sum_{j=0}^{n}\sum_{k=j-1}^{n-1}a_{j,k}Q_k(x)P_j(y)-\sum_{j=0}^{n-1}\sum_{k=0}^{n+1}a_{j,k}Q_k(x)P_j(y)\\
    &=&Q_{n-1}(x)P_n(y)-\sum_{j=0}^{n-1}\sum_{k=n}^{n+1}a_{j,k}Q_k(x)P_j(y)\\
 &=&Q_{n-1}(x)P_n(y)-Q_{n}(x)\sum_{j=0}^{n-1}a_{j,n}P_j(y)-Q_{n+1}(x)\sum_{j=0}^{n-1}a_{j,n+1}P_j(y)\\
 &=&Q_{n-1}(x)P_n(y)-Q_{n}(x)\left(yP_n(y)-a_{n,n}P_n(y)-a_{n+1,n}P_{n+1}(y)\right)\nonumber\\
 &&-Q_{n+1}(x)\left(yP_{n+1}(y)-a_{n,n+1}P_n(y)-a_{n+1,n+1}P_{n+1}(y)\right.\nonumber\\
 &&\left.-a_{n+2,n+1}P_{n+2}(y)\right).\label{eq:sumhatpq1}
\end{eqnarray}
Based on orthogonality property~(\ref{eq:orth}), the needed coefficients in (\ref{eq:sumhatpq1}) correspond to the coefficients in the recurrence relation~(\ref{eq:recurP}). We have
\begin{eqnarray}
&&a_{n+1,n}=1\\
&&a_{n,n}=\frac{c (\alpha +2 n+1)+1}{\beta  c}\\
&&a_{n-1,n}=\frac{n (c (\alpha +n)+2)}{\beta ^2 c}\\
&&a_{n-2,n}=\frac{(n-1) n}{\beta ^3 c}.
\end{eqnarray}
In~(\ref{eq:sumhatpq1}), inserting the recurrence relation~(\ref{eq:recurP}) simplifies the expression and yields the desired result~(\ref{eq:cd}).
% \begin{eqnarray}\label{eq:cdhat}
% (x-y)K(x,y)&=&-Q_{m-1}(x)P_m(y)+Q_{m}(x)\left(a_{m-1,m}P_{m-1}(y)+a_{m-2,m}P_{m-2}(y)\right)\nonumber\\
%  &&+a_{m-1,m+1}Q_{m+1}(x)P_{m-1}(y).
% \end{eqnarray}

As a consequence of~(\ref{eq:cd}), we can write for $x\neq y$,
\begin{equation}
    K(x,y)=\frac{f(x,y)}{g(x,y)},
\end{equation}
where $f(x,y)$ takes the form~(\ref{eq:fp}) and $g(x,y)=x-y$. The limit
\begin{equation}
    K(x,x)=\lim_{y\to x}\frac{f(x,y)}{g(x,y)}
\end{equation}
is evaluated using L'Hôpital's rule, where the required derivatives are obtained from the structure relations~(\ref{eq:structQ}) and~(\ref{eq:structP}). The claimed result~(\ref{eq:cd1}) is deduced after simplifying the resulting expressions using the recurrence relations~(\ref{eq:recurP}) and~(\ref{eq:recurQ}). This completes the proof.
\end{proof}

Proposition~\ref{prop:cdd} states the kernel derivative identity needed in Section \ref{sec:rr} to perform the integration by parts that reveals the relations among spectral moments. We define the operator
\begin{equation}\label{eq:op}
    B_{x,y}=1+x\frac{\dd}{\dd x}+y\frac{\dd}{\dd y}.
\end{equation}

\begin{proposition}\label{prop:cdd}
The correlation kernel~(\ref{eq:K-MOP}) satisfies the derivative formula
\begin{equation}
    B_{x,y}K(x,y)=-\beta  (x-y) K(x,y)+\frac{n }{\beta  c}Q_n(x)P_{n-1}(y) -\beta Q_{n-1}(x) P_n(y).\label{eq:cdd}
\end{equation}
% and
% \begin{equation}
% x\frac{\dd}{\dd x}K(x,x)+K(x,x)=\frac{n }{\beta c}Q_n(x)P_{n-1}(x) -\beta Q_{n-1}(x) P_n(x),\label{eq:cdd1}
% \end{equation}
% respectively.
\end{proposition}
In the special case $y\to x$, we have 
\begin{equation}
x\frac{\dd}{\dd x}K(x,x)+K(x,x)=\frac{n }{\beta c}Q_n(x)P_{n-1}(x) -\beta Q_{n-1}(x) P_n(x).\label{eq:cdd1}
\end{equation}
\begin{proof}
For convenience, we denote 
\begin{equation}\label{eq:f}
f(x,y)=(x-y)K(x,y),
\end{equation}
and it can be seen that
\begin{eqnarray}
x \frac{\dd}{\dd x}K(x,y)+y\frac{\dd}{\dd y}K(x,y)=-K(x,y)+x\frac{\dd}{\dd x}f(x,y)+y\frac{\dd}{\dd y}f(x,y).\label{eq:dkf}
\end{eqnarray}
By the Christoffel-Darboux formula~(\ref{eq:cd}) of correlation kernel, we have 
\begin{eqnarray}\label{eq:fp}
f(x,y)&=&\frac{n (n+1) }{\beta ^3 c}Q_{n+1}(x)P_{n-1}(y)+\frac{n(n-1)  }{\beta ^3 c}Q_n(x)P_{n-2}(y)\nonumber\\
      && +\frac{n (c (\alpha +n)+2) }{\beta ^2 c}Q_n(x)P_{n-1}(y)-Q_{n-1}(x)P_n(y).
\end{eqnarray}
The derivatives of $f(x,y)$ in~(\ref{eq:dkf}) are determined by the structure relations of $Q_n$ and $P_n$ given in Proposition~\ref{propQ} and Proposition~\ref{propstruP},~respectively. The resulting expression is further simplified using the four-point recurrence relations (\ref{eq:recurQ}) and (\ref{eq:recurP}), yielding the compact form~(\ref{eq:cdd}). The special case (\ref{eq:cdd1}) follows directly with $y\to x$. We have completed the proof of Proposition~\ref{prop:cdd}. 
\end{proof}

\subsection{Recurrence relations of spectral moments}\label{sec:rr}

In this section, we derive recurrence relations of spectral moments as defined in (\ref{eq:mk}) by
\begin{equation}
m_k=\mathbb{E}[R_k].
\end{equation}
Note that we work with real $k$ instead of integers such that the spectral moment naturally leads to linear spectral statistics involving logarithmic terms by differentiating with respect to $k$. In our setting, the special case $k=1$ yields the mean value of entanglement entropy as will be shown in the next section, and the general case of positive integer $k$ will serve as the initial data for computing the higher-order cumulants, cf.~\cite{HW25}.

We now discuss the two recurrence relations we obtained for the spectral moments of non-intersecting squared Bessel process, respectively given in Proposition~\ref{prop:r1} and Proposition~\ref{prop:r2}. Proposition~\ref{prop:r1} provides a more compact form that replaces a rather longer six-term recurrence relation purely in $k$ by incorporating derivatives with respect to parameter $T$.

\begin{proposition}\label{prop:r1}
We have 
\begin{eqnarray}
b_3  (k+4) (2 k+7)m_{k+3}&=&b_2  m_{k+2}+b_1  m_{k+1}+ b_0  k (k+1-\alpha) (k+1+\alpha) m_k
\nonumber\\
&&+~\!b_4  \frac{\dd }{\dd T}m_{k+2}+\frac{b_5 }{k+1}\frac{\dd}{\dd T} m_{k+1}. \label{eq:propr1}
\end{eqnarray} 

\end{proposition}

\begin{proof}
By using the one-point correlation function~(\ref{eq:gl}), we have
\begin{equation}\label{eq:m}
m_{k}=\int_{0}^{\infty} x^{k} K(x,x)\dd x .
\end{equation}
Performing the derivative in
\begin{equation}\label{eq:itkk}
 \int_{0}^{\infty} \frac{\dd}{\dd x}x^{k+1} K(x,x)\dd x =0
\end{equation}
by the chain rule, we establish
\begin{equation}\label{eq:itk}
(k+1) m_k =\int_{0}^{\infty} x^{k+1} \frac{\dd}{\dd x} K(x,x)\dd x.
\end{equation}
Making use of the derivative formula~(\ref{eq:cdd1}) in (\ref{eq:itk}) above leads to 
\begin{equation}\label{eq:tk1}
k m_k=\beta {I}_{k}^{n-1,n}-\frac{n}{\beta  c}I_{k}^{n,n-1},
\end{equation}
where we adopt the notation
\begin{equation}\label{eq:Ist}
I_{s,t}(k)=\int_{0}^{\infty} x^{k}Q_s(x)P_t(x)\dd x.    
\end{equation}

We need to recycle the integrals in (\ref{eq:tk1}) into the spectral moments in obtaining the recurrence relation. To this end, we carry out the integrations by parts, in the same spirit as (\ref{eq:itkk}), in $I_{k}^{s,t}$, for $s,t=n-1,n,n+1$, except for $I_{k}^{n-1,n+1}$. The resulting equalities are summarized in (\ref{eq:app31})--(\ref{eq:appa38}) in Appendix~\ref{app3}. These equalities form a linear system when each distinct integral $I_{k}^{s,t}$ is regarded as an individual variable.  By eliminating these variables successively, we obtain the following two compatible relations between $I_{k}^{n,n-1}$ and $I_{k}^{n-1,n}$,
\begin{eqnarray}
 &&\!\!b_6 I_{k-2}^{n-1,n}+b_7 I_{k-1}^{n-1,n}+b_8 I_{k}^{n-1,n}+b_9 I_{k+1}^{n-1,n}+b_{10} I_{k-1}^{n,n-1}+b_{11}I_{k}^{n,n-1}=0, ~~~~~~\label{eq:tk2} \\ 
 &&\!\!b_{12} I_{k-1}^{n-1,n}+b_{13}I_{k}^{n-1,n}+b_{14} I_{k+1}^{n-1,n}+b_{15} I_{k-1}^{n,n-1}+b_{16}I_{k}^{n,n-1}+b_{17} I_{k+1}^{n,n-1}=0.~~~~~~~\label{eq:tk3}
\end{eqnarray}

Continuing, we can further eliminate $I^{n,n-1}$ in (\ref{eq:tk2})--(\ref{eq:tk3}) by using (\ref{eq:tk1}). As a result, we arrive at a relation between integrals $I^{n-1,n}_k$ and the spectral moments $m_k$ of several shifted indices $k$. The claimed result (\ref{eq:propr1}) follows after substituting all the remaining integrals $I^{n-1,n}_k$ with $m_k$ by using the identity
\begin{equation}\label{eq:dbI}
I^{n-1,n}_k=  2 (t-T)^2 \frac{\dd}{\dd T}m_k.
\end{equation}

We now show the derivation of (\ref{eq:dbI}) to complete the proof. Denote $\kappa\!\left(R_k,R\right)$ the joint cumulant (i.e., the covariance) between $R_k$ and $R$, substituting the joint density~(\ref{eq:bi}) into the definition of $\kappa(R_k,R)$, we obtain
\begin{equation}\label{eq:dmRkr}
2 (t-T)^2 \frac{\dd}{\dd T}m_k=\kappa\!\left(R_k,R\right).
\end{equation}
On the other hand, $\kappa\!\left(R_k,R\right)$ admits the following integral form by the two-point correlation function~(\ref{eq:gl}),
\begin{equation}\label{eq:rkr}
 \kappa\!\left(R_k,R\right)=\int_0^{\infty}\!\!\int_0^{\infty} x^k(x-y) K(x,y)K(y,x) \dd x \dd y. 
\end{equation}
With the operator $B_{x,y}$ defined in (\ref{eq:op}) that satisfies the skewed self-adjoint property, we perform the partial integration 
\begin{eqnarray}\label{eq:itkkk}
&&\int_0^{\infty}\!\!\int_0^{\infty} K(x,y)K(y,x) B_{x,y}x^k(x-y)\dd x\dd y \nonumber\\
&=& -\int_0^{\infty}\!\!\int_0^{\infty}  x^k(x-y) B_{x,y}K(x,y)K(y,x)\dd x\dd y.  
\end{eqnarray} 
The integrals on the left-hand side of (\ref{eq:itkkk}) can be rewritten as $\kappa\!\left(R_k,R\right)$ via (\ref{eq:rkr}). On the right-hand side, successive applications of Proposition~\eqref{prop:cdd} and Proposition~\eqref{prop:cd} decouple the double integrals into single ones of form~(\ref{eq:Ist}), and many of them vanish due to the biorthogonal property~(\ref{eq:orth}). Comparing the resulting expression with (\ref{eq:appa36}), we arrive at an interesting relation
\begin{equation}\label{eq:RkRI}
  \kappa\!\left(R_k,R\right)= I_{k}^{n-1,n}.
\end{equation}
Finally, the identity (\ref{eq:dbI}) is proved by inserting (\ref{eq:RkRI}) into (\ref{eq:dmRkr}). 
\end{proof}

We note that the three relations (\ref{eq:tk1}), (\ref{eq:tk2}), and (\ref{eq:tk3}) could yield a recurrence for $m_k$ involving six consecutive indices $k,k+1,\dots,k+5$. We omit it here because the resulting expression is too lengthy to display.  

Proposition~\ref{prop:r2} provides a recurrence relation that requires fewer initial conditions when generating explicit expressions for $m_k$, $k=1,2,\dots$.
 We need the definition
\begin{equation}
    m_{k}^{\pm l} = \mathbb{E}[R_k]\rvert_{n\to n\pm l}.
\end{equation}

\begin{proposition}\label{prop:r2}
 \begin{equation}\label{eq:recurr2}
d_1 m_k= \frac{1}{\beta}\!\left(d_4 m^{-1}_{k-1}-d_3 m_{k-1}-d_2 m^{+1}_{k-1}\right)+\frac{k-1}{\beta} \left( d_7 m_{k-2}+d_6 m^{+1}_{k-2}+d_5 m^{+2}_{k-2}\right).
\end{equation}   
\end{proposition}

\begin{proof}
Similar to the proof of Proposition~\ref{prop:r1}, the recurrence relation~(\ref{eq:recurr2}) is established by appropriately eliminating the variables in~(\ref{eq:app31})--(\ref{eq:appa38}) while keeping in mind the fact 
\begin{eqnarray}
I_{k}^{n,n}=m^{+1}_{k}-m_{k}     
\end{eqnarray}
that follows directly from inserting (\ref{eq:m}) along with the definition of correlation kernel (\ref{eq:orth}) into $m_{k}-m^{-1}_{k}$.
\end{proof}

\section{Dynamical Entanglement Estimation}\label{sec:ee}
In quantum information theory, entanglement of a bipartite pure state are determined
by the eigenvalues of the reduced density matrix of a subsystem~\cite{Mehta,Page93,Foong94,Ruiz95,SZ04,Bortola09,OSZ10,Bortola10,Forrester,Majumdar,Borot12,MajumdarVivo2012,Bortola14,VPO16,FK16,Wei17,SK19,Wei20,Wei20BH,Wei20BHA,LW21,HWC21,Okuyama21,BHK21,BHKRV22,HW22,HW23,HW23C,Wei23C,HW25,WHW25BH}. These eigenvalues are
non-negative and obey the fixed-trace constraint
\begin{equation}\label{eq:l1}
\sum_{i=1}^{n}\lambda_i=1,    
\end{equation}
reflecting the normalization of the density matrix. 

To connect our dynamical eigenvalue model to physically meaningful
entanglement measures that are defined on the simplex of trace-one eigenvalues, we introduce the change of variables
\begin{eqnarray}\label{eq:cvb}
\lambda_i=\frac{x_i}{r},\qquad i=1,\dots,n
\end{eqnarray}
with 
\begin{equation}
    r=\sum_{i=1}^{n}x_i
\end{equation}
so that $0\le\lambda_i\le1$ and (\ref{eq:l1}) is satisfied.
The variable $r$ represents the total trace of the unnormalized spectrum, while
$\boldsymbol\lambda=(\lambda_1,\dots,\lambda_n)$ describes the normalized eigenvalues relevant
for entanglement.

Under this change of variables, the joint density of $(\boldsymbol\lambda,r)$ takes the form
\begin{equation}\label{eq:lamr}
p(\boldsymbol\lambda,r)
=\frac{r^{n-1}}{C}\,
\delta\!\left(1-\sum_{i=1}^n \lambda_i\right)
\det\!\left[f_j(r\lambda_k)\right]_{1\le j,k\le n}
\det\!\left[g_j(r\lambda_k)\right]_{1\le j,k\le n},
\end{equation}
where the factor $r^{n-1}$ arises from the Jacobian of the transformation.
Formally, the corresponding fixed-trace eigenvalue density is obtained by marginalizing over $r$ as
\begin{equation}\label{eq:fbi}
\int_{0}^{\infty} p(\boldsymbol\lambda,r)\,\dd r.
%\qquad \boldsymbol\lambda\in\Delta_{m-1}.
\end{equation}
The two entanglement metrics considered in this work are the quantum purity 
\begin{equation}\label{eq:purity}
P=\sum_{i=1}^n\lambda_i^2,
\end{equation}
and the von Neumann entropy
\begin{equation}\label{eq:von}
S=-\sum_{i=1}^n\lambda_i\ln\lambda_i.
\end{equation}
In the present dynamical model, the joint density~(\ref{eq:lamr}) does not factorize into independent distributions for $r$ and $\boldsymbol\lambda$, preventing the exact moment conversions. To proceed, we replace the random variable $r$ by its mean value $m_1$, where we recall $m_k$ defined in~(\ref{eq:mk}). Under this approximation, entanglement metrics can be expressed in terms of linear statistics of the variables $x_i$. This approximation is motivated by the canonical-microcanonical equivalence principle familiar from statistical mechanics, where a fluctuating extensive constraint is replaced by its typical value when concentration holds~\cite{Ellis1985,Touchette2014}.

Their averages are given by
\begin{equation}\label{eq:approx_obp}
\mathbb{E}[P]
= \frac{m_2}{m_1^2},
\end{equation}
and
\begin{equation}\label{eq:approx_obs}
\mathbb{E}[S]
=\ln m_1-\frac{1}{m_1}\mathbb{E}\!\left[T\right],
\end{equation}
respectively. Therefore, the problem of computing average entanglement metrics now boils down to evaluating $\mathbb{E}[R_2]$ and $\mathbb{E}[T]$ over the non-intersecting squared Bessel process.

\subsection{Computation of average quantum purity}\label{sec:p}

The exact average quantum purity~(\ref{eq:approx_obs}) is given in Proposition~\ref{prop:purity} below.

\begin{proposition}~\label{prop:purity}
The mean value of purity~(\ref{eq:purity}) of the dynamical ensemble~(\ref{eq:fbi}) is 
\begin{equation}
\mathbb{E}[P]=\frac{\alpha ^2 c^2+2 c^2 n^2+3 \alpha  c^2 n+2 \alpha  c+4 c n+1}{n (\alpha  c+c n+1)^2}.
\end{equation}
\end{proposition}
\begin{proof}
According to~(\ref{eq:approx_obs}), it suffices to compute $m_1$ and $m_2$. We utilize the recurrence relations provided in Proposition~\ref{prop:r2}, where the needed initial condition is 
\begin{equation}\label{eq:m0}
    m_0=n.
\end{equation}
Explicitly, setting $k=1$ in Proposition~\ref{prop:r2}, we have
\begin{eqnarray}\label{eq:m10}
 m_1&=&r_1 m_0^{-1}+r_2  m_0+r_3m_0^{+1},
\end{eqnarray}
where
\begin{eqnarray}
 r_1&=&-\frac{n \left(\alpha ^2 c^2-\alpha  c^2+c^2 n^2+2 \alpha  c^2 n-c^2 n+2 \alpha  c+6 c n-c+1\right) }{6 \beta  c (\alpha  c+c n+1)}\label{eq:cr1}\\
 r_2&=&-\frac{n \left(2 \alpha ^2 c^2+\alpha  c^2+2 c^2 n^2+4 \alpha  c^2 n+c^2 n+4 \alpha  c+c+2\right) }{3 \beta  c (\alpha  c+c n+1)}\\
 r_3&=&\frac{n \left(5 \alpha ^2 c^2+\alpha  c^2+5 c^2 n^2+10 \alpha  c^2 n+c^2 n+10 \alpha  c+6 c n+c+5\right) }{6 \beta  c (\alpha  c+c n+1)}\label{eq:cr3}.
\end{eqnarray}
Inserting into the above relation (\ref{eq:m10}) the initial data (\ref{eq:m0}) leads to
\begin{equation}\label{eq:m1}
    m_1=\frac{c n (\alpha +n)+n}{\beta  c}.
\end{equation}
With $m_0$ and $m_1$ known, setting $k=2$ in Proposition~\ref{prop:r2}, we obtain
\begin{equation}\label{eq:m2}
m_2=\frac{c n (\alpha +2 n) (c (\alpha +n)+2)+n}{\beta ^2 c^2}.    
\end{equation}
Putting together (\ref{eq:m1})--(\ref{eq:m2}) in (\ref{eq:approx_obs}), we complete the proof.
\end{proof}

We note that the leading-order behavior of the average purity obtained here,
\begin{equation}\label{eq:HS_purity_asymp}
\mathbb{E}[P]=\frac{2}{n}+O\!\left(n^{-2}\right)
\end{equation}
is consistent with the classical asymptotics~\cite{ZyczkowskiSommers2001} over the Hilbert-Schmidt ensemble.
%providing an a posteriori validation check for the approximation~(\ref{eq:rstar}).

\begin{corollary}\label{coro1}
The limit
\begin{equation}\label{eq:m2_LW}
\lim_{a\to 0} m_2=\frac{n(n+\alpha)(2n+\alpha)}{\beta^2}
\end{equation}
reproduces the average purity~\cite{ZyczkowskiSommers2001} over the Hilbert-Schmidt ensemble after exact moment conversion.
\end{corollary}
   
\begin{proof}
We recall that in the limit $a\to0$, the non-intersecting squared Bessel process is reduced to the Wishart-Laguerre ensemble up to a rescaling~(\ref{eq:gwi}). Under the change of variables (\ref{eq:cvb}), the joint density~(\ref{eq:gwi})
factorizes into the density of Hilbert-Schmidt ensemble
\begin{eqnarray}\label{eq:hs}
f_{\text{HS}}(\bm{\lambda})\propto
\delta\!\left(1-\sum_{i=1}^n \lambda_i\right)
\prod_{i=1}^{n}\lambda_i^\alpha \prod_{1\le i<j\le n}(\lambda_i-\lambda_j)^2,\label{eq:hb}
\end{eqnarray}\label{eq:gamma}
and an independent Gamma density for $r$
\begin{equation}
 h(r)=\frac{\beta^{n(n+\alpha)}}{\Gamma(n(n+\alpha))}\,r^{n(n+\alpha)-1}\e^{-\beta r}.
\end{equation}
The factorization ensures the exact moment conversion, cf.~\cite{Page93,Ruiz95,Wei17,Wei20,HWC21,HW25}, 
\begin{equation}\label{eq:purity_conversion}
\lim_{a\to 0} m_2=\mathbb{E}_{\text {HS}}[P]\mathbb{E}_{h}\!\left[r^2\right]
\end{equation}
with
\begin{equation}\label{eq:Er2}
\mathbb{E}_{h}[r^2]=\frac{n(n+\alpha)\bigl(n(n+\alpha)+1\bigr)}{\beta^2},
\end{equation}
which yields
\begin{equation}\label{eq:HS_purity}
\mathbb{E}_{\text{HS}}[P]
=\frac{2n+\alpha}{n(n+\alpha)+1},
\end{equation}
in agreement with the known average purity formula over Hilbert-Schmidt ensemble~\cite{ZyczkowskiSommers2001}. The proof is completed.
\end{proof}

\subsection{Computation of average entanglement entropy}\label{sec:v}
The task of computing average entanglement entropy boils down to compute $\mathbb{E}[T]$, where we need the inverse moment $m_{-1}$ provided in Lemma~\ref{lemmaim} as one of the initial conditions.

\begin{lemma}\label{lemmaim}
 The inverse moment $m_{-1}$ is given by
\begin{eqnarray}
    m_{-1}&=&-\frac{n \beta }{\alpha c(\alpha +  n)}\, _1F_1\left(n+1;n+\alpha +1;-\frac{1}{c}\right) \, _1F_1\left(1-n;-n-\alpha +1;\frac{1}{c}\right)\nonumber\\
    &&+\frac{\beta (\alpha +n)}{\alpha } \, _1F_1\left(-n;-n-\alpha ;\frac{1}{c}\right) \, _1F_1\left(n;n+\alpha ;-\frac{1}{c}\right)-\beta.\label{eq:m-1result}
\end{eqnarray}   
\end{lemma}
Here, ${}_1F_1$ denotes the confluent hypergeometric function
\begin{equation}
{}_1F_1(a;b;z)
  = \sum_{j=0}^{\infty} \frac{(a)_j}{(b)_j}\,\frac{z^j}{j!}.
\end{equation}
\begin{proof}
As an application of Proposition~\ref{prop:cdd}, setting $k=-1$ in (\ref{eq:tk1}) gives
\begin{equation}
m_{-1}=-\beta I_{-1}^{n-1,n}+\frac{n}{\beta  c}I_{-1}^{n,n-1}.
\end{equation} 
Due to orthogonality, many terms cancel if we expand the polynomials $P_{n-1}$ and $P_n$ in the above integrals, resulting in
\begin{eqnarray}
m_{-1}&=&-\beta^{1-n}\Gamma (n+1)  L_{n}^{(-n-\alpha -1)}\!\left(\frac{1}{c}\right)\int_0^{\infty}x^{-1}Q_{n-1}(x)\dd x-\beta \nonumber\\
&&+\frac{n }{ c}\beta^{-n}\Gamma (n) L_{n-1}^{(-n-\alpha)}\!\left(\frac{1}{c}\right)\int_0^{\infty}x^{-1}Q_{n}(x) \dd x,\label{eq:m-1}
\end{eqnarray}
where the first and second integrals arise from the constant terms in $P_{n}(x)$ and $P_{n-1}(x)$, respectively, while the term $-\beta$ is contributed by the highest-order term $x^n$ in $P_{n}$ as
\begin{equation}
 -\beta \int_0^{\infty} x^{n-1} Q_{n-1}(x) \dd x=-\beta.
\end{equation}
The result (\ref{eq:m-1}) then follows after applying the integral identity 
\begin{eqnarray}
\int_0^{\infty } x^{s-1} \widetilde{Q}_n^{(\alpha,c)}(x) \dd x &=&\frac{(-1)^n \e^{1/c} \Gamma (s) \Gamma (s+\alpha )}{ c^{s+\alpha}\Gamma (n+\alpha +1) \Gamma (s-n)}\nonumber\\
&&\times ~_1F_1\left(n-s+1;\alpha +n+1;-\frac{1}{c}\right),\label{eq:xqk} 
\end{eqnarray}
which is derived by inserting the Rodrigues' formula (\ref{eq:rdq}) into the left-hand-side of~(\ref{eq:xqk}), and performing integration by parts $n$ times. 
% the resulting integral can be evaluated as~\cite{CV03}
% \begin{equation}
% \int_0^{\infty } x^{s-n-1} w_{\alpha +n,c}(x) \, dx=\frac{(-1)^n \e^{1/c}  \Gamma (s+\alpha ) }{c^{\alpha+s}\Gamma (n+\alpha +1)}\, _1F_1\left(n-s+1;n+\alpha +1;-\frac{1}{c}\right),
% \end{equation}
\end{proof}

% We continue to compute the needed initial conditions $m_{-1}$ and $\mathbb{E}[T_0]$ as according to Proposition~\ref{prop:ET}.

% \begin{lemma}\label{lemmaxQ}
% For real $s$ such that the integral is well-defined, we have
% \begin{eqnarray}
% \int_0^{\infty } x^{s-1} \widetilde{Q}_n^{(\alpha,c)}(x) \dd x &=&\frac{(-1)^n \e^{1/c} \Gamma (s) \Gamma (s+\alpha )}{ c^{s+\alpha}\Gamma (n+\alpha +1) \Gamma (s-n)}\nonumber\\
% &&\times ~_1F_1\left(n-s+1;\alpha +n+1;-\frac{1}{c}\right).\label{eq:xqk} 
% \end{eqnarray}
% \end{lemma}
% \begin{proof}
% Inserting the Rodrigues' formula (\ref{eq:rdq}) into the left-hand-side of~(\ref{eq:xqk}), and performing integration by parts $n$ times, the resulting integral can be evaluated as~\cite{CV03}
% \begin{equation}
% \int_0^{\infty } x^{s-n-1} w_{\alpha +n,c}(x) \, dx=\frac{(-1)^n \e^{1/c}  \Gamma (s+\alpha ) }{c^{\alpha+s}\Gamma (n+\alpha +1)}\, _1F_1\left(n-s+1;n+\alpha +1;-\frac{1}{c}\right),
% \end{equation}
% where ${}_1F_1$ denotes the confluent hypergeometric function
% \begin{equation}
% {}_1F_1(a;b;z)
%   = \sum_{j=0}^{\infty} \frac{(a)_j}{(b)_j}\,\frac{z^j}{j!}.
% \end{equation}
% This completes the proof.
% \end{proof}

% Lemma~\ref{lemmaim} provides an exact formula of the initial condition $m_{-1}$, which is also referred to as inverse moment.

% As a consequence of the above explicit expression for $\mathbb{E}[T]$, we obtain the average of dynamical entanglement entropy under the approximation~(\ref{eq:approx_obs}) in Proposition~\ref{prop:ee} below.
\begin{proposition}\label{prop:ee}
The mean value of entanglement entropy~(\ref{eq:von}) of the dynamical ensemble~(\ref{eq:fbi}) is 
\begin{equation}\label{eq:eS}
\mathbb{E}[S]= \ln \frac{c n (\alpha +n)+n}{\beta  c}-\frac{\beta  c}{c n (\alpha +n)+n}\mathbb{E}\!\left[T\right].
\end{equation}
\end{proposition}
An explicit expression of $\mathbb{E}\!\left[T\right]$ is given by
\begin{eqnarray}\label{eq:ETe}
    \mathbb{E}[T]&=&\frac{e^{-1/c}}{\beta }\left(n\sum _{k=0}^{n-3} \frac{(-1)^k c^{\alpha +k+2} \Gamma (-k+n-2)}{\Gamma (k+1)}\Phi_{\alpha ,\alpha }^{k+1,n-1} L_{n-k}^{(-n-\alpha -1)}\left(\frac{1}{c}\right)\right.\nonumber\\
    &&-\sum _{k=0}^{n-2} \frac{c^{\alpha +k+1} \left((-1)^k \Gamma (-k+n-1)\right)}{\Gamma (k+1)}\Phi _{\alpha ,\alpha }^{k+1,n} L_{-k+n-1}^{(-n-\alpha) }\left(\frac{1}{c}\right)\nonumber\\
    &&-(-1)^n n\sum _{k=0}^2 \frac{c^{k+\alpha +n} L_{2-k}^{(-n-\alpha -1)}\left(\frac{1}{c}\right)}{\Gamma (k+1) \Gamma (k+n-1)}\Bigg((\psi_0(k+n)-\psi_0(k+1)) \Phi _{\alpha ,\alpha }^{k+n-1,n-1}\nonumber\\
    &&+\frac{\dd \left(\Phi _{v,v}^{k+n-1,n-1}-\Phi _{v,\alpha }^{k+n-1,n-1}\right)}{\dd v}\Bigg)\!-\frac{(-1)^n c^{\alpha +n}}{ \Gamma (n)}\Bigg((\psi_0(n+1)+\gamma ) \nonumber\\
    &&\times \Phi _{\alpha ,\alpha }(n,n)+\!\left.\!\left.\frac{\dd \left(\Phi _{v,v}(n,n)-\Phi _{v,\alpha }(n,n)\right)}{\dd v}\Bigg)\right)\right\rvert_{v\to\alpha}+\frac{ n (\alpha  c+c n+1)}{\beta  c}\nonumber\\
   && \times \left(\log \left(\frac{c}{\beta }\right)-1\right),
\end{eqnarray}
where 
\begin{eqnarray}
\Phi_{v_1,v_2}^{k,n}&=& \frac{(-1)^n  c^{-k-v_2-1} \Gamma (k+1)\Gamma (k+v_2+1) }{\Gamma (n+v_1+1)}\nonumber\\
&&\times \, _1F_1\left(k+v_2+1;n+v_1+1;\frac{1}{c}\right),
\end{eqnarray}
the function $\psi_0(z)$ denotes the digamma function, defined by
\begin{equation}
\psi_0(z)=\frac{\mathrm{d}}{\mathrm{d}z}\ln\Gamma(z),
\end{equation}
and
\begin{equation}
\psi_0(1)=-\gamma
\end{equation}
with $\gamma$ being the Euler gamma.
\begin{proof}
The result (\ref{eq:eS}) is obtained by substituting the expression (\ref{eq:m1}) of $m_1$ into (\ref{eq:approx_obs}). The explicit expression of (\ref{eq:ETe}) is obtained as follows.

 By taking derivative of the recurrence relation (\ref{eq:r2}) given in Proposition~\ref{prop:r2} before setting $k=1$, we arrive at
    \begin{eqnarray}\label{eq:t10}
    \mathbb{E}\!\left[T\right]&=&r_1 \mathbb{E}^{+}\!\left[T_0\right]+r_2 \mathbb{E}\!\left[T_0\right]+r_3 \mathbb{E}^{-}\!\left[T_0\right]+r_0,
\end{eqnarray}
where the constants $r_1$, $r_2$, $r_3$ are given in (\ref{eq:cr1})--(\ref{eq:cr3}) and $r_0$ is
\begin{eqnarray}
r_0&=&-\frac{\alpha }{6 \beta ^2 c (\alpha  c+c n+1) (\alpha  c+c n+c+1)}\left(\alpha ^4 c^4+\alpha ^3 c^4-\alpha ^2 c^4-\alpha  c^4+c^4 n^4\right.\nonumber\\
&&+4 \alpha  c^4 n^3+c^4 n^3+6 \alpha ^2 c^4 n^2+3 \alpha  c^4 n^2-c^4 n^2+4 \alpha ^3 c^4 n+3 \alpha ^2 c^4 n-c^4 n\nonumber\\
&&-2 \alpha  c^4 n+4 \alpha ^3 c^3+\alpha ^2 c^3-4 \alpha  c^3+12 c^3 n^3+28 \alpha  c^3 n^2+13 c^3 n^2+20 \alpha ^2 c^3 n\nonumber\\
&&+14 \alpha  c^3 n-c^3+6 \alpha ^2 c^2-\alpha  c^2+18 c^2 n^2+28 \alpha  c^2 n+7 c^2 n-3 c^2+4 \alpha  c\nonumber\\
&&\!\left.+12 c n-c+1\right)m_{-1}+\frac{2 \alpha  c n\!\left(m_{-1}^{+1}-(n+1)m_{-1}^{+2}\right)}{3 \beta ^2 (\alpha  c+c n+1) (\alpha  c+c n+c+1)}\nonumber\\
&&+\frac{n}{6 (\beta +\beta  c (\alpha +n))}\left(12 n+\alpha +c^2 \left(\alpha ^3-\alpha +n^3+3 \alpha  n^2+3 \alpha ^2 n-n\right)\right.\nonumber\\
&&\!\left.+c \left(2 \alpha ^2-4 \alpha +13 n^2+15 \alpha  n-1\right)\right).
\end{eqnarray}
The expression of $m_{-1}$ is given in Lemma~\ref{lemmaim}, hence $r_0$. To compute $\mathbb{E}\!\left[T_0\right]$, we utilize the Christoffel-Darboux formula~(\ref{eq:cdc}) in the one-point correlation function and expand the involved type II polynomials by their explicit expressions converted from (\ref{eq:pklkx}). The resulting integrals are then evaluated by the following identity established from the derivative of (\ref{eq:xqk}),
\begin{eqnarray}
    \int_0^{\infty}x^{k}\ln x~\!Q_n(x)\dd x&=&\frac{\psi_0(k+1)-\psi_0(k-n+1)}{\Gamma(k+1-n)} \Phi _{\alpha ,\alpha }^{k,n}\nonumber\\
    &&+\!\left.\frac{1}{\Gamma(k+1-n)}\frac{\dd \left(\Phi_{v,v}^{k,n}-\Phi _{v,\alpha }^{k,n}\right)}{\dd v}\right\rvert_{v\to\alpha}.
\end{eqnarray}
The proof is completed after putting together the needed initial conditions in (\ref{eq:t10}).\\
\end{proof}

\begin{corollary}\label{coroT}
The limit 
\begin{equation}
    \lim_{a\to0}\mathbb{E}[T]=\frac{n (\alpha +n) }{\beta }\psi_0(n+\alpha )-\frac{n  (\alpha +n)}{\beta }\ln \beta +\frac{n (n+1)}{2 \beta}
\end{equation}
reproduces the average entanglement entropy over the Hilbert-Schmidt ensemble after exact moment conversion.
 \end{corollary}
\begin{proof}
Similar to the purity case, by performing the moment conversion
\begin{equation}
\lim_{a\to0}\mathbb{E}[T]= \mathbb{E}_h\!\left[r\ln r\right]-\mathbb{E}_h\!\left[r\right] \mathbb{E}_{\text{HS}}\!\left[S\right]
\end{equation}
with
\begin{eqnarray}
\mathbb{E}_h\!\left[r\right]&=&\frac{n (n+\alpha)}{\beta  },\\
\mathbb{E}_h\!\left[r\ln r\right]&=&\frac{n (n+\alpha) }{\beta }(\psi_0(n (n+\alpha)+1)-\ln \beta ),
\end{eqnarray}
we recover the average entanglement entropy over the Hilbert-Schmidt ensemble as
\begin{equation}
\mathbb{E}_{\text {HS}}\!\left[S\right]=\psi_0(n (n+\alpha )+1)-\psi_0(n+\alpha )-\frac{n+1}{2 (\alpha +n)},
\end{equation}
with $n$ and $n+\alpha$ being the dimensions of the smaller and larger subsystems, respectively, in Page’s bipartite setting~\cite{Page93}.    
\end{proof}

\section*{Acknowledgment}
The authors wish to express their sincere gratitude to Makoto Katori for his hospitality during their visit to Chuo University, where the present work was initiated. Lu Wei also acknowledges the support by the U.S. National Science Foundation (2306968) and the U.S. Department of Energy (DE-SC0024631).

\appendix
\section{Coefficients and identities}
In the appendices, we list the explicit coefficients underlying the shorthand notations used in main text, together with several identities established and invoked in the proof.
\subsection{Coefficients in Proposition~\ref{prop:r1}}
\begin{eqnarray}
 b_0&=& \frac{2 \alpha  (4 k+13)-c (k+4) \left(-\alpha ^2+k^2+5 k-\alpha  n+6\right)}{\beta ^2},\\
 b_1&=&\frac{1}{\beta  c}\left(156 \alpha -12 \alpha ^3 c^2+48 \alpha  c^2+3 \alpha  c^2 k^4+7 c^2 k^4 n+27 \alpha  c^2 k^3+69 c^2 k^3 n\right.\nonumber\\
 &&-3 \alpha ^3 c^2 k^2+84 \alpha  c^2 k^2-4 \alpha  c^2 k^2 n^2-7 \alpha ^2 c^2 k^2 n+242 c^2 k^2 n-15 \alpha ^3 c^2 k.\nonumber\\
 &&+108 \alpha  c^2 k-22 \alpha  c^2 k n^2-37 \alpha ^2 c^2 k n+360 c^2 k n-24 \alpha  c^2 n^2-36 \alpha ^2 c^2 n.\nonumber\\
 &&+192 c^2 n-24 \alpha ^2 c+11 c k^4+81 c k^3-11 \alpha ^2 c k^2-28 \alpha  c k^2 n+202 c k^2.\nonumber\\
 &&\!\left.-41 \alpha ^2 c k-130 \alpha  c k n+192 c k-132 \alpha  c n+48 c+32 \alpha  k^2+152 \alpha  k\right)\!,
  \end{eqnarray}
 \begin{eqnarray}
 b_2&=&\frac{3 }{c}\big(c^2 \left(k^2+6 k+8\right) \left(-\alpha ^2+k^2+6 k-m^2-2 \alpha  m+9\right)+2 c (k+4)\nonumber\\
 &&\times(5 \alpha +2 \alpha  k+4 k m+11 m)+2 \left(10 k^2+57 k+80\right)\!\big),\\
 b_3&=&-6 \beta,\\
 b_4&=&2 (t-T)^2 \left(-24 \alpha  c+\alpha  c k^4+3 c k^4 n+6 \alpha  c k^3+27 c k^3 n-\alpha ^3 c k^2+5 \alpha  c k^2\right.\nonumber\\
 &&-\alpha ^2 c k^2 n+84 c k^2 n-4 \alpha ^3 c k-18 \alpha  c k-4 \alpha ^2 c k n+108 c k n+48 c n+24 k^4\nonumber\\
 &&\!\left.+198 k^3-8 \alpha ^2 k^2+582 k^2-26 \alpha ^2 k+720 k+312\right)\!,\\
 b_5&=&-\frac{2 (t-T)^2\beta}{c (k+2)}\left(c \left(24 \alpha +4 \alpha  k^2+21 k^2 n+19 \alpha  k+108 k n+132 n\right)-c^2 (k+4)\right.\nonumber\\
 &&\times\left(-3 \alpha ^2+2 k^3+16 k^2-2 \alpha ^2 k-3 k n^2-5 \alpha  k n+42 k-6 n^2-9 \alpha  n+36\right)\nonumber\\
 &&\!\left.-6 (k+2) (4 k+13)\right)\!,\\
b_6&=&3 \alpha  \beta  c^3 (k-1) k^2 (k+1) \left((k-1)^2-\alpha ^2\right)\!,\\
b_7&=&\beta ^2 c^2 k (k+1) \left(-2 \alpha ^3-\alpha ^4 c+2 c k^4-2 c k^3+5 \alpha ^2 c k^2+13 \alpha  c k^2 n-\alpha ^2 c k\right.\nonumber\\
&&\!\left.-6 \alpha  c k n-\alpha ^3 c n-4 \alpha  k^2+3 \alpha  k\right)\!,\\
b_8&=&\beta ^3 c (k+1) \left(-4 \alpha +\alpha ^3 c^2+\alpha  c^2 k^3+8 c^2 k^3 n-2 \alpha  c^2 k^2-4 c^2 k^2 n+2 \alpha ^3 c^2 k\right.\nonumber\\
&&+\alpha  c^2 k+4 \alpha  c^2 k n^2+6 \alpha ^2 c^2 k n-4 c^2 k n+2 \alpha  c^2 n^2+3 \alpha ^2 c^2 n+10 c k^3\nonumber\\
&&\!\left.-2 c k^2+\alpha  c k n-8 c k+2 \alpha  c n-5 \alpha  k\right)\!,\\
b_9&=&-\beta ^4 c^2 k (k+2) \left(\alpha ^2 c-\alpha +2 c k^2-2 c k+\alpha  c n\right)\!,\\
b_{10}&=&c k (k+1) n \left(2 \alpha ^3+\alpha ^4 c+4 c k^4-4 c k^3-5 \alpha ^2 c k^2-7 \alpha  c k^2 n+4 \alpha ^2 c k+3 \alpha  c k n\right.\nonumber\\
&&\!\left.+\alpha ^3 c n-2 \alpha  k^2\right)\!,\\
b_{11}&=&-\beta  (k+1) n \left(-4 \alpha +2 c^2 k^3 n+2 c^2 k^2 n+\alpha  c^2 k n^2+\alpha ^2 c^2 k n-4 c^2 k n\right.\nonumber\\
&&\!\left.+2 \alpha  c^2 n^2+2 \alpha ^2 c^2 n-2 \alpha ^2 c+16 c k^3-8 c k^2-4 \alpha ^2 c k+7 \alpha  c k n-8 c k\right.\nonumber\\
&&\!\left.+2 \alpha  c n-8 \alpha  k\right)\!,\\
b_{12}&=&\beta ^2 c^2 k (k+1) \left(k^2-\alpha ^2\right) \left(\alpha -\alpha ^2 c+c k^2+c k-\alpha  c n\right)\!,\\
b_{13}&=&\beta ^3 c (k+1) \left(-2 \alpha -\alpha ^3 c^2+2 \alpha  c^2 k^3+4 c^2 k^3 n+3 \alpha  c^2 k^2+6 c^2 k^2 n-2 \alpha ^3 c^2 k\right.\nonumber\\
&&\!\left.+\alpha  c^2 k-4 \alpha  c^2 k n^2-6 \alpha ^2 c^2 k n+2 c^2 k n-2 \alpha  c^2 n^2-3 \alpha ^2 c^2 n+3 \alpha ^2 c+5 c k^3\right.\nonumber\\
&&\!\left.+9 c k^2+3 \alpha ^2 c k+8 \alpha  c k n+4 c k+4 \alpha  c n-4 \alpha  k\right)\!,\\
b_{14}&=&\beta ^4 \left(-c^2\right) k (k+2) \left(4 \alpha -\alpha ^2 c+c k^2+c k-\alpha  c n\right)\!,\\
b_{15}&=&c k (k+1) n \left(k^2-\alpha ^2\right) \left(-\alpha +\alpha ^2 c+2 c k^2+2 c k+\alpha  c n\right)\!,\\
b_{16}&=&\beta  (-(k+1)) n \left(-2 \alpha +3 \alpha  c^2 k^3+c^2 k^3 n+6 \alpha  c^2 k^2+3 c^2 k^2 n-3 \alpha ^3 c^2 k\right.\nonumber\\
&&\!\left.+3 \alpha  c^2 k-4 \alpha  c^2 k n^2-7 \alpha ^2 c^2 k n+2 c^2 k n-2 \alpha  c^2 n^2-2 \alpha ^2 c^2 n+2 \alpha ^2 c\right.\nonumber\\
&&\!\left.+8 c k^3+12 c k^2+4 \alpha ^2 c k+8 \alpha  c k n+4 c k+4 \alpha  c n-4 \alpha  k\right)\!, \\
b_{17}&=&6 \alpha  \beta ^2 c k (2 k+3) n.
\end{eqnarray}

\subsection{Relations of integrals $I_{s,t}$}
\begin{eqnarray}\label{eq:app31}
&&n(-1+c(-k+n))\beta\,\mathrm{I}^{n-1,n-1}_{k-1}
- cn\beta^{2}\mathrm{I}^{n-1,n-1}_{k}
- (1+c+c\alpha)\beta^{2}\mathrm{I}^{n-1,n}_{k-1}\nonumber\\
&&
+ n^{2}(1+c(n+\alpha))\mathrm{I}^{n-1,n}_{k-1}
+ c\beta^{3}\mathrm{I}^{n-1,n}_{k}
+ n(1+c(1+2n+\alpha))\beta\,\mathrm{I}^{n,n}_{k-1}
- cn\beta^{2}\mathrm{I}^{n,n}_{k}\nonumber\\
&&
- c\beta^{3}\mathrm{I}^{n+1,n-1}_{k-1}
+ cn\beta^{2}\mathrm{I}^{n+1,n-1}_{k}=0,
\end{eqnarray}
\begin{eqnarray}
&&n\beta\mathrm{I}^{n-1,n-1}_{k-1}
+ (1+c(1-k+2n+\alpha))\beta^{2}\mathrm{I}^{n-1,n}_{k-1}
- c\beta^{3}\mathrm{I}^{n-1,n}_{k}
\nonumber\\
&&+ n(1+c(n+\alpha))\beta\mathrm{I}^{n,n}_{k-1}
+ n(1+n)\mathrm{I}^{n,n+1}_{k-1}
+ c\beta^{3}\mathrm{I}^{n+1,n-1}_{k-1}=0,
\end{eqnarray}
\begin{eqnarray}
&&n(1+n)\mathrm{I}^{n-1,n-1}_{k-1}
+ (1+n)(1+c(1+n+\alpha))\beta\mathrm{I}^{n-1,n}_{k-1}
+ c(1+k)\beta^{2}\mathrm{I}^{n+1,n-1}_{k-1}\nonumber\\
&&
- n(1+c(n+\alpha))\beta\mathrm{I}^{n+1,n}_{k-1}
- n(1+n)\mathrm{I}^{n+1,n+1}_{k-1}=0,
\end{eqnarray}
\begin{eqnarray}
&&cn\beta\mathrm{I}^{n-1,n-1}_{k-1}
+ c\beta^{2}\mathrm{I}^{n-1,n}_{k-1}
+ n(2+c(k+n+\alpha))\mathrm{I}^{n-1,n}_{k-1}
\nonumber\\
&&+ (2+c(2+3n+2\alpha))\beta\mathrm{I}^{n,n}_{k-1}
- 2c\beta^{2}\mathrm{I}^{n,n}_{k}
+ 2c\beta^{2}\mathrm{I}^{n+1,n}_{k-1}=0,
\end{eqnarray}
\begin{eqnarray}
c(k-1)\beta\mathrm{I}^{n,n}_{k-1}
+ c\beta^{2}\mathrm{I}^{n-1,n}_{k-1}
- n\mathrm{I}^{n-1,n}_{k-1}
+ (1+n)\mathrm{I}^{n,n+1}_{k-1}
- c\beta^{2}\mathrm{I}^{n+1,n}_{k-1}=0,
\end{eqnarray}
\begin{eqnarray}\label{eq:appa36}
&&n(1+n)\mathrm{I}^{n-1,n}_{k-1}
+ (1+n)(1+c(1+n+\alpha))\beta\mathrm{I}^{n,n}_{k-1}
+ c\beta^{3}\mathrm{I}^{n+1,n-1}_{k-1}\nonumber\\
&&
+ (1+c(1+k+2n+\alpha))\beta^{2}\mathrm{I}^{n+1,n}_{k-1}
- c\beta^{3}\mathrm{I}^{n+1,n}_{k}
+ (1+n)\beta\mathrm{I}^{n+1,n+1}_{k-1}=0,~~~
\end{eqnarray}
\begin{eqnarray}
&&-(2+c+3cn+2c\alpha)\beta\,\mathrm{I}^{n,n}_{k-1}
+2c\beta^{2}\mathrm{I}^{n,n}_{k}
-2c\beta^{2}\mathrm{I}^{n-1,n}_{k-1}-c\beta^{2}\mathrm{I}^{n+1,n}_{k-1}\nonumber\\
&&
-(1+n)(2+c(1-k+n+\alpha))\mathrm{I}^{n,n+1}_{k-1}
-c(1+n)\beta\,\mathrm{I}^{n+1,n+1}_{k-1}=0,
\end{eqnarray}
\begin{eqnarray}\label{eq:appa38}
&&-c\beta^{2}\mathrm{I}^{n-1,n}_{k-1}
-(1+c(1+2n+\alpha))\beta\,\mathrm{I}^{n,n}_{k-1}
+c\beta^{2}\mathrm{I}^{n,n}_{k}
-(1+n)\nonumber\\
&&\times(1+c(1+n+\alpha))\mathrm{I}^{n,n+1}_{k-1}
+\frac{c\beta^{3}}{1+n}\mathrm{I}^{n+1,n-1}_{k-1}
+\frac{(1+c(\alpha-1))\beta^{2}}{1+n}\mathrm{I}^{n+1,n}_{k-1}
\nonumber\\
&&-\frac{c\beta^{3}}{1+n}\mathrm{I}^{n+1,n}_{k}
+(1-c(1+k+n))\beta\,\mathrm{I}^{n+1,n+1}_{k-1}
+c\beta^{2}\mathrm{I}^{n+1,n+1}_{k}=0.
\end{eqnarray}

\subsection{Coefficients in Proposition~\ref{prop:r2}}\label{app3}

\begin{eqnarray}
d_1&=&c (k+1) (2 k+1) (c k-\alpha  c-c n-c-1) \left(-2 \alpha  c^2+c^2 k^3-3 c^2 k^2 n-5 c^2 k^2\right.\nonumber\\
&&-\alpha ^2 c^2 k+\alpha  c^2 k+2 c^2 k n^2+\alpha  c^2 k n+9 c^2 k n+8 c^2 k-2 c^2 n^2-2 \alpha  c^2 n\nonumber\\
&&\!\left.-6 c^2 n-4 c^2+2 \alpha  c-3 c k^2-3 \alpha  c k+5 c k-2 c-2 k+2\right)\!,
\end{eqnarray}
\begin{eqnarray}
d_2&=&2 (k-1) (k+1) (k-n-1) (k-n-\alpha -2) (k-n-\alpha -1)^2 (k-n-\alpha ) \nonumber\\
&&\times(k+n+\alpha ) c^5+(k-n-\alpha -1) \left(2 k^6-15 n k^5-14 \alpha  k^5-18 k^5+58 n^2 k^4\right.\nonumber\\
&&+22 \alpha ^2 k^4+94 n k^4+77 n \alpha  k^4+62 \alpha  k^4+42 k^4-63 n^3 k^3-10 \alpha ^3 k^3-192 n^2 k^3\nonumber\\
&&-77 n \alpha ^2 k^3-58 \alpha ^2 k^3-157 n k^3-130 n^2 \alpha  k^3-233 n \alpha  k^3-82 \alpha  k^3+18 n^4 k^2\nonumber\\
&&-34 k^3+112 n^3 k^2+15 n \alpha ^3 k^2+14 \alpha ^3 k^2+166 n^2 k^2+48 n^2 \alpha ^2 k^2+112 n \alpha ^2 k^2\nonumber\\
&&+38 \alpha ^2 k^2\!+78 n k^2+51 n^3 \alpha  k^2+210 n^2 \alpha  k^2+176 n \alpha  k^2+30 \alpha  k^2+4 k^2-12 n^4 k\nonumber\\
&&-21 n^3 k-n \alpha ^3 k+2 \alpha ^3 k-14 n^2 \alpha ^2 k+9 n \alpha ^2 k+10 \alpha ^2 k+12 n k-25 n^3 \alpha  k\nonumber\\
&&-14 n^2 \alpha  k+18 n \alpha  k+8 \alpha  k+4 k-6 n^4-28 n^3-4 n \alpha ^3-6 \alpha ^3-32 n^2-14 n^2 \alpha ^2\nonumber\\
&&-32 n \alpha ^2-12 \alpha ^2-12 n-16 n^3 \alpha\!\left. -54 n^2 \alpha -36 n \alpha -4 \alpha \right) c^4-\left(14 k^6-92 n k^5\right.\nonumber\\
&&-46 \alpha  k^5-70 k^5+148 n^2 k^4+50 \alpha ^2 k^4+339 n k^4+192 n \alpha  k^4+174 \alpha  k^4+128 k^4\nonumber\\
&&-120 n^3 k^3-18 \alpha ^3 k^3-401 n^2 k^3-144 n \alpha ^2 k^3-134 \alpha ^2 k^3-422 n k^3-264 n^2 \alpha  k^3\nonumber\\
&&-510 n \alpha  k^3-228 \alpha  k^3-104 k^3+14 n^4 k^2+188 n^3 k^2+44 n \alpha ^3 k^2+84 n^2 \alpha ^2 k^2\nonumber\\
&&+306 n^2 k^2+30 \alpha ^3 k^2+231 n \alpha ^2 k^2+106 \alpha ^2 k^2+179 n k^2+54 n^3 \alpha  k^2+385 n \alpha  k^2\nonumber\\
&&+395 n^2 \alpha  k^2+112 \alpha  k^2+34 k^2-12 n^4 k-38 n^3 k-8 n \alpha ^3 k-6 \alpha ^3 k-22 n^2 \alpha ^2 k\nonumber\\
&&-13 n^2 k-31 n \alpha ^2 k-10 \alpha ^2 k+14 n k-26 n^3 \alpha  k-59 n^2 \alpha  k-25 n \alpha  k-6 \alpha  k-2k\nonumber\\
&&-2 n^4-30 n^3-6 n \alpha ^3-6 \alpha ^3-40 n^2-10 n^2 \alpha ^2-32 n \alpha ^2-12 \alpha ^2-18 n-56 n^2 \alpha\nonumber\\
&& \!\left.-6 n^3 \alpha -40 n \alpha -6 \alpha \right) c^3+\left(14 k^5-35 n k^4-28 \alpha  k^4-52 k^4+27 n^2 k^3+14 \alpha ^2 k^3\right.\nonumber\\
&&+64 n k^3+38 n \alpha  k^3+78 \alpha  k^3+72 k^3+30 n^3 k^2-14 n^2 k^2-63 n \alpha ^2 k^2-26 \alpha ^2 k^2\nonumber\\
&&-33 n k^2-15 n^2 \alpha  k^2-89 n \alpha  k^2-70 \alpha  k^2-44 k^2-20 n^3 k-7 n^2 k+21 n \alpha ^2 k\nonumber\\
&&+10 \alpha ^2 k+10 n k-5 n^2 \alpha  k+35 n \alpha  k+18 \alpha  k+10 k-10 n^3-6 n^2+12 n \alpha ^2\nonumber\\
&&+2 \alpha ^2-6 n-\!\left.2 n^2 \alpha +4 n \alpha +2 \alpha \right) c^2-2 \left(2 k^4+10 n k^3-2 \alpha  k^3-6 k^3-8 n^2 k^2\right.\nonumber\\
&&+22 n \alpha  k^2-9 n k^2+4 \alpha  k^2+6 k^2+6 n^2 k-3 n k-11 n \alpha  k-2 \alpha  k-2 k\nonumber\\
&&+\!\left.2 n^2+2 n-6 n \alpha \right) c-4 (k-1) (3 k+1) n,
\end{eqnarray}
\newpage
\begin{eqnarray}
d_3&=&\frac{1}{2} (k-1) (k+1) (k-n-\alpha -2) (k-n-\alpha -1) (k-n-\alpha ) (k+n+\alpha ) \nonumber\\
&&\times\left(k^2+2 n k+\alpha  k+k-2 n^2-4 n-2 n \alpha -2 \alpha -2\right) c^5+\frac{1}{2} \left(2 k^7-24 n k^6\right.\nonumber\\
&&-12 \alpha  k^6-25 k^6+52 n^2 k^5+8 \alpha ^2 k^5+102 n k^5+54 n \alpha  k^5+49 \alpha  k^5+55 k^5\nonumber\\
&&-8 n^3 k^4+12 \alpha ^3 k^4-121 n^2 k^4+34 n \alpha ^2 k^4+11 \alpha ^2 k^4-126 n k^4+14 n^2 \alpha  k^4\nonumber\\
&&-82 n \alpha  k^4-44 \alpha  k^4-33 k^4-42 n^4 k^3-10 \alpha ^4 k^3-76 n^3 k^3-78 n \alpha ^3 k^3-49 \alpha ^3 k^3\nonumber\\
&&+15 n^2 k^3-168 n^2 \alpha ^2 k^3-210 n \alpha ^2 k^3-57 \alpha ^2 k^3+26 n k^3-142 n^3 \alpha  k^3\nonumber\\
&&-237 n^2 \alpha  k^3-76 n \alpha  k^3-9 \alpha  k^3-9 k^3+20 n^5 k^2+108 n^4 k^2+14 n \alpha ^4 k^2\nonumber\\
&&+14 \alpha ^4 k^2+162 n^3 k^2+62 n^2 \alpha ^3 k^2+124 n \alpha ^3 k^2+44 \alpha ^3 k^2+107 n^2 k^2\nonumber\\
&&+102 n^3 \alpha ^2 k^2+314 n^2 \alpha ^2 k^2+228 n \alpha ^2 k^2+39 \alpha ^2 k^2+38 n k^2+74 n^4 \alpha  k^2\nonumber\\
&&+312 n^3 \alpha  k^2+346 n^2 \alpha  k^2+142 n \alpha  k^2+16 \alpha  k^2+10 k^2-12 n^5 k-30 n^4 k\nonumber\\
&&-2 n \alpha ^4 k+2 \alpha ^4 k-32 n^3 k-18 n^2 \alpha ^3 k+8 n \alpha ^3 k+9 \alpha ^3 k-35 n^2 k-42 n^3 \alpha ^2 k\nonumber\\
&&-20 n^2 \alpha ^2 k+22 n \alpha ^2 k+9 \alpha ^2 k-16 n k-38 n^4 \alpha  k-56 n^3 \alpha  k-19 n^2 \alpha  k-10 n \alpha  k\nonumber\\
&&-8 n^5-36 n^4-4 n \alpha ^4-6 \alpha ^4-46 n^3-20 n^2 \alpha ^3-42 n \alpha ^3-16 \alpha ^3-36 n^3 \alpha ^2\nonumber\\
&&-18 n^2-102 n^2 \alpha ^2-70 n \alpha ^2-10 \alpha ^2-28 n^4 \alpha -102 n^3 \alpha -\!\left.100 n^2 \alpha -28 n \alpha \right) c^4\nonumber\\
&&-\frac{1}{2} \left(27 k^6-64 n k^5-33 \alpha  k^5-87 k^5+35 n^2 k^4-27 \alpha ^2 k^4+156 n k^4+20 n \alpha  k^4\right.\nonumber\\
&&+44 \alpha  k^4+87 k^4+106 n^3 k^3+33 \alpha ^3 k^3+53 n^2 k^3+166 n \alpha ^2 k^3+87 \alpha ^2 k^3-92 n k^3\nonumber\\
&&+239 n^2 \alpha  k^3+132 n \alpha  k^3+5 \alpha  k^3-25 k^3-32 n^4 k^2-216 n^3 k^2-38 n \alpha ^3 k^2\nonumber\\
&&-44 \alpha ^3 k^2-185 n^2 k^2-72 n^2 \alpha ^2 k^2-252 n \alpha ^2 k^2-79 \alpha ^2 k^2-32 n k^2-66 n^3 \alpha  k^2\nonumber\\
&&-388 n^2 \alpha  k^2-236 n \alpha  k^2-24 \alpha  k^2-2 k^2+24 n^4 k+66 n^3 k+10 n \alpha ^3 k+3 \alpha ^3 k\nonumber\\
&&+67 n^2 k+32 n^2 \alpha ^2 k+22 n \alpha ^2 k+9 \alpha ^2 k+32 n k+46 n^3 \alpha  k+73 n^2 \alpha  k+48 n \alpha  k\nonumber\\
&&+8 \alpha  k+8 n^4+44 n^3+4 n \alpha ^3+8 \alpha ^3+30 n^2+8 n^2 \alpha ^2+40 n \alpha ^2+10 \alpha ^2+12 n^3 \alpha \nonumber\\
&&+\!\left.68 n^2 \alpha +32 n \alpha \right) c^3+\left(12 k^5-14 n k^4+12 \alpha  k^4-23 k^4-47 n^2 k^3-24 \alpha ^2 k^3\right.\nonumber\\
&&+16 n k^3-65 n \alpha  k^3-30 \alpha  k^3+17 k^3-14 n^3 k^2+48 n^2 k^2+25 n \alpha ^2 k^2+33 \alpha ^2 k^2\nonumber\\
&&+n k^2-7 n^2 \alpha  k^2+92 n \alpha  k^2+25 \alpha  k^2-7 k^2+8 n^3 k+n^2 k-9 n \alpha ^2 k-6 \alpha ^2 k\nonumber\\
&&-6 n k+5 n^2 \alpha  k-12 n \alpha  k-6 \alpha  k+k+6 n^3-2 n^2-4 n \alpha ^2-3 \alpha ^2+3 n-9 n \alpha\nonumber\\
&&+\!\left.6 n^2 \alpha  -\alpha \right) c^2+\left(4 k^4-16 n k^3-16 \alpha  k^3-7 k^3-17 n^2 k^2+19 n k^2+16 n \alpha  k^2\right.\nonumber\\
&&\left.+23 \alpha  k^2+4 k^2+12 n^2 k-2 n k-7 n \alpha  k-6 \alpha  k-k+5 n^2-n-5 n \alpha -\alpha \right) c\nonumber\\
&&-2 (k-1) \left(2 k^2-2 n k-k-n\right)\!,
\end{eqnarray}
\newpage
\begin{eqnarray}
d_4&=&\frac{1}{2} n \left(-2 \alpha  c^2+c^2 k^3-3 c^2 k^2 n-5 c^2 k^2-\alpha ^2 c^2 k+\alpha  c^2 k+2 c^2 k n^2+\alpha  c^2 k n\right.\nonumber\\
&&+9 c^2 k n+8 c^2 k-2 c^2 n^2-2 \alpha  c^2 n-6 c^2 n-4 c^2+2 \alpha  c-3 c k^2-3 \alpha  c k+5 c k\nonumber\\
&&\!\left.-2 c-2 k+2\right) \left(\alpha ^2 c^2+\alpha  c^2+c^2 k^3-2 \alpha  c^2 k^2-2 c^2 k^2 n+\alpha ^2 c^2 k-\alpha  c^2 k\right.\nonumber\\
&&+c^2 k n^2+2 \alpha  c^2 k n-c^2 k n-c^2 k+c^2 n^2+2 \alpha  c^2 n+c^2 n+\alpha  c-3 c k^2\nonumber\\
&&+3 \alpha  c k+9 c k n+c k\!\left.+3 c n+2 k\right)\!,\\
d_5&=&c^4 (k-1) (k+1) (n+1) (k-\alpha -1) (k-\alpha -n-1) (k-\alpha -n) (\alpha +k+n)\nonumber\\
&&+c^3 (n+1) (k-\alpha -1) \left(-2 \alpha ^2-\alpha +2 k^4-6 \alpha  k^3-3 k^3 n-5 k^3+4 \alpha ^2 k^2\right.\nonumber\\
&&+7 \alpha  k^2+10 k^2 n^2+23 \alpha  k^2 n+5 k^2 n+2 k^2-2 \alpha ^2 k-6 k n^2-11 \alpha  k n-k n+k\nonumber\\
&&-\!\left.4 n^2-8 \alpha  n-n\right)-c^2 (k-1) (n+1) \left(\alpha ^2+\alpha +5 k^3-10 \alpha  k^2-11 k^2 n-6 k^2\right.\nonumber\\
&&\!\left.+5 \alpha ^2 k+5 \alpha  k+23 \alpha  k n+8 k n+k+7 \alpha  n+3 n\right)\nonumber\\
&&+2 c (k^2-k)  (n+1) (k-\alpha -1),\\
d_6&=&(k-1) (k+1) (k-\alpha -1) (k-n-\alpha -2) (k-n-\alpha -1) (k-n-\alpha ) \nonumber\\
&&\times(k+\alpha -1) (k+n+\alpha ) c^5+(k-1) (k-\alpha -1) \left(k^5-8 n k^4-7 \alpha  k^4-10 k^4\right.\nonumber\\
&&+24 n^2 k^3+5 \alpha ^2 k^3+32 n k^3+29 n \alpha  k^3+19 \alpha  k^3+13 k^3-16 n^3 k^2+7 \alpha ^3 k^2\nonumber\\
&&-53 n^2 k^2+10 n \alpha ^2 k^2+2 \alpha ^2 k^2-30 n k^2-13 n^2 \alpha  k^2-45 n \alpha  k^2-8 \alpha  k^2-2 k^2\nonumber\\
&&-n^4 k-6 \alpha ^4 k+8 n^3 k-31 n \alpha ^3 k-7 \alpha ^3 k+17 n^2 k-45 n^2 \alpha ^2 k-32 n \alpha ^2 k+\alpha ^2 k\nonumber\\
&&+4 n k-21 n^3 \alpha  k-17 n^2 \alpha  k+8 n \alpha  k-2 k-n^4-4 \alpha ^4+4 n^3-17 n \alpha ^3-6 \alpha ^3\nonumber\\
&&+14 n^2-23 n^2 \alpha ^2-18 n \alpha ^2+2 \alpha ^2+6 n-11 n^3 \alpha \!\left.-8 n^2 \alpha +12 n \alpha +2 \alpha \right) c^4\nonumber\\
&&-\left(8 k^6-30 n k^5-16 \alpha  k^5-33 k^5+26 n^2 k^4-6 \alpha ^2 k^4+114 n k^4+27 n \alpha  k^4\right.\nonumber\\
&&+38 \alpha  k^4+50 k^4+5 n^3 k^3+28 \alpha ^3 k^3-59 n^2 k^3+72 n \alpha ^2 k^3+37 \alpha ^2 k^3-159 n k^3\nonumber\\
&&+34 n^2 \alpha  k^3-31 n \alpha  k^3-33 \alpha  k^3-36 k^3-14 \alpha ^4 k^2-7 n^3 k^2-69 n \alpha ^3 k^2-50 \alpha ^3 k^2\nonumber\\
&&+37 n^2 k^2-72 n^2 \alpha ^2 k^2-173 n \alpha ^2 k^2-33 \alpha ^2 k^2+85 n k^2-17 n^3 \alpha  k^2-75 n^2 \alpha  k^2\nonumber\\
&&-n \alpha  k^2+19 \alpha  k^2+14 k^2+8 \alpha ^4 k-n^3 k+42 n \alpha ^3 k+10 \alpha ^3 k-n^2 k+44 n^2 \alpha ^2 k\nonumber\\
&&+49 n \alpha ^2 k-3 \alpha ^2 k+n k+10 n^3 \alpha  k+16 n^2 \alpha  k-n \alpha  k-7 \alpha  k-3 k+6 \alpha ^4+3 n^3\nonumber\\
&&+27 n \alpha ^3+12 \alpha ^3-3 n^2+28 n^2 \alpha ^2+48 n \alpha ^2+5 \alpha ^2-11 n+7 n^3 \alpha +25 n^2 \alpha +6 n \alpha\nonumber\\
&&-\alpha \big) c^3+(k-1) \left(5 k^4-8 n k^3+6 \alpha  k^3-4 k^3-21 n^2 k^2-27 \alpha ^2 k^2-8 n k^2\right.\nonumber\\
&&-53 n \alpha  k^2-23 \alpha  k^2+3 k^2+16 \alpha ^3 k+14 n^2 k+61 n \alpha ^2 k+23 \alpha ^2 k+12 n k\nonumber\\
&&+33 n^2 \alpha  k+70 n \alpha  k+5 \alpha  k-4 k+4 \alpha ^3+7 n^2+19 n \alpha ^2+6 \alpha ^2+4 n+27 n \alpha \nonumber\\
&&\!\left.+11 n^2 \alpha +2 \alpha \right) c^2+(k-1) (k-\alpha -1) \left(3 k^2-19 n k-9 \alpha  k-3 k-5 n-\alpha \right) c \nonumber\\
&&-2 (k-1) k (k-\alpha -1),
\end{eqnarray}
\begin{eqnarray}
d_7&=&-\frac{1}{2} k (k+1) (k-\alpha -1) (k-n-\alpha -2) (k-n-\alpha -1) (k-n-\alpha ) (k+\alpha -1) \nonumber\\
&&\times (k+n+\alpha ) c^5+\frac{1}{2} (k-\alpha -1) (k+\alpha -1)\left(k^5+5 n k^4+5 \alpha  k^4+5 k^4\right.\nonumber\\
&&-25 n^2 k^3-13 \alpha ^2 k^3-23 n k^3-38 n \alpha  k^3-17 \alpha  k^3-8 k^3+19 n^3 k^2+7 \alpha ^3 k^2\nonumber\\
&&+35 n^2 k^2+33 n \alpha ^2 k^2+9 \alpha ^2 k^2+12 n k^2+45 n^2 \alpha  k^2+44 n \alpha  k^2+2 \alpha  k^2+7 n^3 k\nonumber\\
&&+3 \alpha ^3 k+18 n^2 k+13 n \alpha ^2 k+8 \alpha ^2 k+6 n k+17 n^2 \alpha  k+26 n \alpha  k+2 \alpha  k-2 n^3\nonumber\\
&&-2 \alpha ^3-2 n^2-6 n \alpha ^2-2 \alpha ^2\!\left.-6 n^2 \alpha -4 n \alpha \right) c^4+\frac{1}{2} \left(13 k^6-42 n k^5-32 \alpha  k^5\right.\nonumber\\
&&-50 k^5+57 n^2 k^4+6 \alpha ^2 k^4+139 n k^4+80 n \alpha  k^4+85 \alpha  k^4+64 k^4-10 n^3 k^3\nonumber\\
&&+32 \alpha ^3 k^3-149 n^2 k^3+50 n \alpha ^2 k^3-7 \alpha ^2 k^3-161 n k^3+8 n^2 \alpha  k^3-166 n \alpha  k^3\nonumber\\
&&-71 \alpha  k^3-28 k^3-19 \alpha ^4 k^2+18 n^3 k^2-88 n \alpha ^3 k^2-29 \alpha ^3 k^2+109 n^2 k^2\nonumber\\
&&-65 n^2 \alpha ^2 k^2-67 n \alpha ^2 k^2+8 \alpha ^2 k^2+63 n k^2-14 n^3 \alpha  k^2-32 n^2 \alpha  k^2+72 n \alpha  k^2\nonumber\\
&&+13 \alpha  k^2-k^2+\alpha ^4 k-6 n^3 k+10 n \alpha ^3 k-7 \alpha ^3 k+n^2 k+11 n^2 \alpha ^2 k-13 n \alpha ^2 k\nonumber\\
&&+\alpha ^2 k+11 n k+8 n^3 \alpha  k+8 n^2 \alpha  k+22 n \alpha  k+7 \alpha  k+2 k+6 \alpha ^4-2 n^3+22 n \alpha ^3\nonumber\\
&&+6 \alpha ^3-18 n^2+18 n^2 \alpha ^2+16 n \alpha ^2-2 \alpha ^2-10 n+6 n^3 \alpha +\!\left.16 n^2 \alpha -8 n \alpha -2 \alpha \right) c^3\nonumber\\
&&-\frac{1}{2} (k-\alpha -1) \left(13 k^4-41 n k^3+12 \alpha  k^3-24 k^3-20 n^2 k^2-25 \alpha ^2 k^2+58 n k^2\right.\nonumber\\
&&-93 n \alpha  k^2-7 \alpha  k^2+15 k^2+12 n^2 k+11 \alpha ^2 k-9 n k+43 n \alpha  k-3 \alpha  k-4 k+8 n^2\nonumber\\
&&+\!\left.6 \alpha ^2-8 n+26 n \alpha \right) c^2-(k-\alpha -1) \left(2 k^3-17 n k^2-8 \alpha  k^2-k+12 n k\right.\nonumber\\
&&\!\left.+6 \alpha  k-k^2+5 n+\alpha \right) c+2 (k-1) k (k-\alpha -1).
\end{eqnarray}

\end{document}